\documentclass[11pt]{article}
\usepackage{graphicx, multicol}
\usepackage{comment}
\usepackage{amsmath}

\usepackage{mathrsfs} 
\usepackage[a4paper, left=0.8in, right=0.8in]{geometry} 

\usepackage[latin1]{inputenc}
\usepackage{mathtools}
\usepackage{color}
\usepackage{orcidlink}
\usepackage{amsthm,amsmath}
\usepackage{amssymb}
\usepackage{amsfonts}
\usepackage{hyperref}
\hypersetup{colorlinks,allcolors=black}
\usepackage{graphicx}
\usepackage{tabularx}
\usepackage{lscape}
\usepackage{pdflscape}

\usepackage{adjustbox} 
\usepackage{array} 

\newtheorem{theorem}{Theorem}
\newtheorem{definition}{Definition}
\newtheorem{corollary}{Corollary}
\newtheorem{proposition}{Proposition}
\newtheorem{lemma}{Lemma}
\newtheorem{example}{Example}
\newtheorem{remark}{Remark}

\usepackage{authblk}

\date{ }

\begin{document}

\title{\textbf{Quasi-recursive MDS Matrices over Galois Rings}}

\author{
\begin{tabular}{@{}c@{}}
\textbf{Shakir Ali}$^{1,2}$\orcidlink{0000-0001-5162-7522}, 
\textbf{Atif Ahmad Khan}$^{1}$ \orcidlink{0009-0009-2371-5707}, 
\textbf{Abhishek Kesarwani}$^{3}$\orcidlink{
	0000-0001-8488-6811}, 
\textbf{Susanta Samanta}$^{4}$ \orcidlink{0000-0003-4643-5117} \\
\small
$^{1}$Department of Mathematics, Faculty of Science \\ Aligarh Muslim University, Aligarh 202002, India\\
$^{2}$Faculty of Mathematics and Natural Sciences, Universitas Gadjah Mada, Yogyakarta 55281, Indonesia \\
$^{3}$Department of Mathematics \& Computing, Indian Institute of Information Technology Vadodara, International Campus Diu, Daman \& Diu 362520, India \\
$^{4}$Department of Electrical and Computer Engineering, University of Waterloo\\ Waterloo, Ontario N2L 3G1, Canada \\
\texttt{shakir.ali.mm@amu.ac.in, atifkhanalig1997@gmail.com, abhishek\_kesarwani@iiitvadodara.ac.in, ssamanta@uwaterloo.ca}
\end{tabular}
}

\maketitle

\begin{abstract}


Let $p$ be a prime and $s,m,n$ be positive integers. This paper studies quasi-recursive MDS matrices over Galois rings $GR(p^{s}, p^{sm})$ and proposes various direct construction methods for such matrices. The construction is based on skew polynomial rings $GR(p^{s}, p^{sm})[X;\sigma]$, whose rich factorization properties and enlarged class of polynomials are used to define companion matrices generating quasi-recursive MDS matrices. First, two criteria are established for characterizing polynomials that yield recursive MDS matrices, generalizing existing results, and then an additional criterion is derived in terms of the right roots of the associated Wedderburn polynomial. Using these criteria, methods are developed to construct skew polynomials that give rise to quasi-recursive MDS matrices over Galois rings. This framework extends known constructions to the non-commutative setting and significantly enlarges the family of available matrices, with potential applications to efficient diffusion layers in cryptographic primitives. The results are particularly relevant for practical implementations when $s = 1$ and $p = 2$, i.e., over the finite field $\mathbb{F}_{2^m}$, which is of central interest in real-world cryptographic applications.

\end{abstract}


\section{Introduction}
The notions of confusion and diffusion, originally introduced by Shannon~\cite{1949ShannonCommunication}, form the cornerstone of modern symmetric-key cryptographic design. In such designs, the round function typically integrates a non-linear layer to ensure confusion and a linear layer to achieve diffusion. This work focuses on the construction of linear diffusion layers that maximize the propagation of internal dependencies. Perfect diffusion can be characterized in several equivalent ways, one approach employs multipermutations~\cite{Vaudenay1995on}, while another defines it through Maximum Distance Separable (MDS) matrices~\cite{Daemen2002the}.

Because MDS matrices achieve the optimal branch number, they have become a central component in the diffusion layers of many block ciphers and hash functions. Over the years, numerous methods have been proposed for constructing MDS matrices, which can broadly be classified into two categories: non-recursive and recursive constructions. In non-recursive constructions, the resulting matrices are directly MDS. In contrast, recursive constructions begin with a sparse matrix $A$ of order $n$, carefully chosen such that $A^n$ becomes an MDS matrix. In this setting, we refer to $A$ as a recursive MDS matrix.

Recursive MDS matrices offer notable implementation advantages, especially in lightweight cryptographic settings. Their recursive structure allows the diffusion layer to be realized efficiently by repeatedly applying the same sparse matrix, thereby reducing hardware complexity and requiring only a few clock cycles. This property has been effectively utilized in practical designs such as the \textsc{PHOTON}~\cite{Guo2011the} family of hash functions and the \textsc{LED}~\cite{Guo2011TheLED} block cipher, where companion matrix based recursive MDS constructions enable efficient implementations through simple LFSR. several special structures for recursive MDS matrices have been introduced, including GFS~\cite{Recursive_Diffusion_Layer12}, DSI~\cite{DSI}, and DLS~\cite{Rec_MDS_2022}.

In the existing literature, extensive efforts have been devoted to the direct algebraic construction of MDS matrices in both recursive and non-recursive settings. The non-recursive constructions are predominantly based on classical Cauchy and Vandermonde structures~\cite{chand2013constructions, gupta2019cryptographically,lacan2004systematic,sajadieh2012construction}. In contrast, recursive constructions are typically derived through coding theoretic techniques. For instance, in 2014, Augot and Finiasz~\cite{augot2014direct} proposed a method for constructing recursive MDS matrices based on shortened BCH codes. In their approach, a polynomial \( g(X) \) is first computed whose roots are consecutive powers of some element \( \beta \) in an extension field of \( \mathbb{F}_q \). If \( g(X) \in \mathbb{F}_q[X] \), then the BCH code generated by \( g(X) \) becomes MDS. The final step is to verify whether the computed polynomial indeed belongs to \( \mathbb{F}_q[X] \); if so, the polynomial \( g(X) \) yields a recursive MDS matrix. Similarly, Berger~\cite{berger2013construction} employed Gabidulin codes to design recursive MDS matrices. Subsequently, in a series of works~\cite{Guptadirect, gupta2017towards,gupta2019almost}, several efficient methods were proposed for constructing recursive MDS matrices using companion matrices defined over finite fields, and it was shown that an infinite class of recursive MDS matrices can be obtained through this approach.
Kesarwani et al.~\cite{Kesarwani2021} further studied recursive MDS matrices over finite commutative rings, but the literature on recursive MDS matrices over finite rings remains relatively limited.

To address this, we extend the recursive construction to the framework of skew polynomial ring over a Galois ring. Related work in this direction over the finite field $\mathbb{F}_{2^{2m}}$ has been carried out by Cauchois et al.~\cite{Cauchois2016direct}. Skew polynomial rings constitute a significant class of non-commutative rings and have been utilized in the construction of algebraic codes~\cite{Boulagouaz2013codes,Bhaintwal2012Skew}, where codes are defined as ideals within quotient rings or modules over skew polynomial rings. The main motivation for exploring MDS matrices in this framework lies in the rich factorization properties of polynomials in skew polynomial rings, which lead to a greater number of ideals compared to the commutative case.


In this paper, we describe the direct constructions of quasi-recursive MDS matrices over Galois rings. Let \(C_g\) be the companion matrix associated with a monic skew polynomial \(g(X) \in GR(p^s, p^{sm})[X;\sigma]\) of degree \(k\). A quasi-recursive MDS matrix is an MDS matrix that can be expressed as the product of companion matrices, namely \(C_g^{[k-1]} \cdot C_g^{[k-1]} \cdots C_g^{[1]} \cdot C_g\). Here, \(C_g^{[i]}\) denotes the companion matrix of the skew polynomial \(g(X)\) with each entry transformed by the automorphism \(\sigma^i\). We first establish two criteria for characterizing polynomials that generate quasi-recursive MDS matrices and show that, under suitable conditions, the resulting matrix is involutory. We then formulate an additional criterion based on the right roots of the associated Wedderburn polynomial and, using this, propose several methods to construct skew polynomials that yield quasi-recursive MDS matrices.

 
Our approach is more general, as it operates over the skew polynomial ring $GR(p^s, p^{sm})[X;\sigma]$, yielding a larger class of polynomials and corresponding right roots that give rise to quasi-recursive MDS matrices, as established in Lemma 3. This generalization arises from the flexibility in the choice of the ring and the automorphism. In particular, if we take $s = 1$ and $\sigma$ as the identity map, our construction reduces to the classical work due to Gupta et al. \cite[Theorem 1]{Guptadirect} and \cite[Theorem 1]{gupta2017towards}, respectively. However, when $\sigma$ is chosen as a non-identity automorphism, the calculations are performed in a non-commutative setting, providing a broader framework that encompasses new classes of quasi-recursive MDS matrices. 

This study is intended to contribute to the theoretical development of quasi-recursive MDS matrices over Galois rings, while also retaining practical relevance. In the special case when $s = 1$ and $p = 2$, the ring $GR(p^s, p^{sm})[X; \sigma]$ reduces to $\mathbb{F}_{2^{m}}[X; \sigma]$, which has notable practical relevance. This is highlighted by the work of Cauchois et al.~\cite{Cauchois2016direct}, who proposed a hardware architecture for quasi-recursive matrices and performed the computations using a generalized form of the LFSR, known as the Skew Linear Feedback Shift Register (SLFSR). Consequently, our study not only extends the theoretical developments of such constructions but also holds potential for efficient practical implementations.

\vspace{1em}

\noindent This paper is organized as follows. Section~\ref{Sec:Pre} introduces key definitions and results on linear codes and MDS codes. Section~\ref{Sec:skew-poly-pre} reviews skew polynomial rings, and Section~\ref{Sec:Skew-cyclic-code} examines skew cyclic codes and their role in constructing quasi-recursive MDS matrices. Section~\ref{Sec:polynomial-charac} presents the key theoretical results on quasi-recursive MDS matrices, which are then utilized in Section~\ref{Sec:MDS-simple-roots} to derive several direct constructions of quasi-recursive MDS matrices over Galois rings. Finally, Section~\ref{Sec:Conclusion} concludes the paper.

\section{Preliminaries}\label{Sec:Pre}
In this section, we present some basic definitions and preliminary results that will be useful in establishing our subsequent findings. Let $\mathcal{R}$ be a finite commutative ring with unity. The set of all vectors of length $n$ with entries from $\mathcal{R}$ is denoted by $\mathcal{R}^n$. We denote  $\mathcal{U}(\mathcal{R})$ and $\mathcal{N}(\mathcal{R})$ the set of all units  of $\mathcal{R}$, and set of all nilpotent elements of $\mathcal{R}$, respectively. Furthermore, $\mathcal{M}_k(\mathcal{R})$ represents the set of all $k \times k$ matrices over $\mathcal{R}$, and $I_k$ denotes the identity matrix of order $k$. We begin this section with the definition of a linear code.

  \begin{definition}
  
A linear code $\mathcal{C}$ over $\mathcal{R}$ of length $n$ is a submodule of $\mathcal{R}^n$.

\end{definition}

Let $\textbf{u} = (u_0,u_1,\ldots,u_{n-1})$ and $\textbf{v} = (v_0,v_1,\ldots,v_{n-1})$ be two codewords in $\mathcal{C}$. The Hamming distance between $u$ and $v$
is defined as
\[
d_H(\textbf{u},\textbf{v}) = |\{i : u_i \neq v_i\}|,
\]
and the minimum Hamming distance $d(\mathcal{C})$ of a code $\mathcal{C}$ is defined as
\[
d(\mathcal{C}) = \min\{d_H(\textbf{u},\textbf{v})\mid \textbf{u},\textbf{v} \in \mathcal{C},\ \textbf{u} \neq \textbf{v}\}.
\]
The notation $[n,k,d]$ is commonly used to denote a linear code $\mathcal{C}$ over a finite field, where $n$, $k$, and $d$ represent the length, dimension, and minimum distance of the code, respectively. 
For codes over rings, this notation is not always meaningful, since the dimension may not be defined for all codes. 
However, when $\mathcal{R}^n$ is a free module over $\mathcal{R}$, the notation $[n,k,d]$ is still used.

In coding theory, it is desirable to construct linear codes that achieve large values of the information rate $k/n$ and the relative minimum distance $d/n$. However, there exists an inherent trade-off among the parameters $n$, $k$, and $d$. In particular, the classical Singleton bound provides a  upper limit on the minimum distance attainable by a code.

\begin{theorem}\cite[Theorem 1.1]{MacWilliams1977the}
Let $\mathcal{C}$ be a linear code over $\mathcal{R}$ with minimum Hamming distance $d$. Then,
\[
d \leq n-k+1.
\]
\end{theorem}

\begin{definition}
A linear code $\mathcal{C}$ over $\mathcal{R}$ with parameter $[n,k,d]$ is said to be MDS (Maximum Distance Separable) if $d = n-k+1$.
\end{definition}

\begin{definition}
For a set $W = \{\textbf{w}_1,\textbf{w}_2,\ldots,\textbf{w}_k\}$ of elements of a free $\mathcal{R}$-module $\mathcal{R}^n$, if a relation
\[
\sum_{i=1}^{k} a_i \textbf{w}_i = 0,
\]
with $a_i \in \mathcal{R}$ and $\textbf{w}_i \in \mathcal{R}^n$ implies that all the coefficients $a_i$ are $0$, we say that the elements of $W$ are linearly independent
(over $\mathcal{R}$); otherwise they are linearly dependent.
\end{definition}
In 1997, Dong et al. \cite{Dong} provided a matrix characterization of MDS codes over modules. Specifically, let $\mathcal{C}$  be a $[n,k,d]$ linear code over $\mathcal{R}$ with parity check matrix $\mathcal{H}$ of the form $\mathcal{H} = (B_{n-k\times k} \mid I_{n-k})$. Then, $\mathcal{C}$ is an MDS code if and only if the determinant of every $t \times t$ submatrix of $B$, where $t \in \{1,2,\dots,\min\{k,n-k\}\}$, belongs to $\mathcal{U}(\mathcal{R})$. In otherwords, a square matrix $M$ of order $k$ over $\mathcal{R}$ is an MDS matrix if and only if every square 
submatrices of $M$ has determinant in $\mathcal{U}(\mathcal{R})$.

\begin{lemma}\cite{MacWilliams1977the}\label{lem3}
Let $\mathcal{C}$ be a $[n,k,d]$ linear code over $\mathcal{R}$. Then, the following are equivalent: 
\begin{itemize}
    \item[\textnormal{(i)}] $\mathcal{C}$ is an MDS code.
    \item[\textnormal{(ii)}] Every set of $k$ columns of a generator matrix $\mathcal{G}$ is linearly independent.
    \item[\textnormal{(iii)}] Every set of $n-k$ columns of a parity-check matrix $\mathcal{H}$ is linearly independent.
\end{itemize}
\end{lemma}

As a direct consequence, we have the following result:

\begin{corollary}\cite[Corollary, p.319]{MacWilliams1977the}\label{Fact1}
Let $M$ be a $k \times k$ matrix over $\mathcal{R}$. Then, $M$ is MDS if and only if any $k$ columns 
of 
\(
\mathcal{G} = [M \mid I_k]
\)
are linearly independent over $\mathcal{R}$. Equivalently, $M$ is MDS if and only if any $k$ rows of
\(
\bar{\mathcal{G}} =\Big[ \frac{I_k}{M}\Big]
\)
are linearly independent over $\mathcal{R}$. Moreover, if $M$ is MDS, then so is $-M$. 
\end{corollary}
The rest of this section, detailed in the next subsection, introduces skew polynomial rings over Galois rings and discusses their fundamental properties.

\subsection{Skew polynomials ring over Galois ring}\label{Sec:skew-poly-pre}

Let $p$ be a prime number and  $s, m$ be positive integers. Consider a monic basic primitive polynomial $f(Y) \in \mathbb{Z}_{p^s}[Y]$ of degree $m$. The Galois ring of characteristic $p^s$ with $p^{sm}$ elements, denoted by $GR(p^s, p^{sm})$, is defined as the quotient ring
\[
GR(p^s, p^{sm}) = \mathbb{Z}_{p^s}[Y]/\langle f(Y) \rangle,
\]
where $\langle f(Y) \rangle$ denotes the ideal generated by $f(Y)$. This is a local ring with maximal ideal $\langle p \rangle$, and its residue field is
\[
GR(p^s, p^{sm}) / \langle p \rangle \cong \mathbb{F}_{p^m}.
\]
The natural projection $GR(p^s, p^{sm}) \to \mathbb{F}_{p^m}$ is denoted by $a \mapsto \bar{a}$, and extends naturally to a homomorphism from $GR(p^s, p^{sm})[X]$ to $\mathbb{F}_{p^m}[X]$. Let $\zeta$ be a primitive root of $f(X)$ in $GR(p^s, p^{sm})$. The associated Teichm$\ddot{u}$ller set is given by
\[
\mathcal{T}_m = \{0, 1, \zeta, \dots, \zeta^{p^m-2}\}.
\]
Each element $c \in GR(p^s, p^{sm})$ admits a unique representation of the form
\[
c = c_0 + c_1 \zeta + \cdots + c_{m-1} \zeta^{m-1}, \qquad c_i \in \mathbb{Z}_{p^s}.
\]

Let $\sigma : GR(p^s, p^{sm}) \longrightarrow GR(p^s, p^{sm})$ be an automorphism. The skew polynomial ring over $GR(p^s, p^{sm})$ with respect to $\sigma$, denoted by $GR(p^s, p^{sm})[X;\sigma]$, is the set
\[
GR(p^s, p^{sm})[X;\sigma] = \left\{ a_0 + a_1X + \cdots + a_{n-1}X^{\,n-1} : a_i \in GR(p^s, p^{sm}), \; n \in \mathbb{N} \right\},
\]
equipped with the usual polynomial addition and a multiplication governed by the rule
\[
aX^i\ast bX^j = a\sigma^i(b)X^{i+j},
\quad \text{for all } a,b \in GR(p^s, p^{sm}).
\]

\noindent The center of the skew polynomial ring $GR(p^s, p^{sm})[X;\sigma]$, denoted by 
$Z\big(GR(p^s, p^{sm})[X;\sigma]\big)$, is defined as
\[
Z\big(GR(p^s, p^{sm})[X;\sigma]\big)
= \left\{\, f(X) \in GR(p^s, p^{sm})[X;\sigma]  :~ f(X)\ast g(X) = g(X) \ast f(X) \text{ for all } g(X) \in GR(p^s, p^{sm})[X;\sigma] \,\right\}.
\]

Following \cite{BoucherSkew}, we collect some properties of the skew polynomial ring $GR(p^s,p^{sm})[X;\sigma]$:

\begin{theorem}
Let $\sigma$ be an automorphism of $GR(p^s, p^{sm})$, and let $t$ denote its order. Then,
\begin{itemize}
    \item[\textnormal{(i)}] The ring $GR(p^s, p^{sm})[X;\sigma]$ is non-commutative unless $\sigma = \mathrm{id}$, where $\mathrm{id}$ denotes the identity map on $GR(p^s,p^{sm})$ .
    \item[\textnormal{(ii)}]   If $f(X)\ast h(X)\in Z(GR(p^s, p^{sm})[X; \sigma]) $, then $f(X)\ast h(X)=h(X)\ast f(X).$
    \item[\textnormal{(iii)}] The center of the ring is given by
    \[
    Z(GR(p^s, p^{sm})[X;\sigma]) = GR(p^s, p^{sm})^{\sigma}[X^t],
    \]
    where $GR(p^s, p^{sm})^{\sigma}$ is the set of all elements fixed by $\sigma$.
\end{itemize}
\end{theorem}

The Frobenius automorphism on $GR(p^s, p^{sm})$ over $\mathbb{Z}_{p^s}$, denoted by $\theta$, is defined by
\[
\theta\!\left(\sum_{i=0}^{n-1} c_i \zeta^i \right) = \sum_{i=0}^{n-1} c_i \zeta^{ip}.
\]
The set of all automorphisms of $GR(p^s, p^{sm})$ forms a group, called the automorphism group, and denoted by $\mathrm{Aut}(GR(p^s, p^{sm}))$. This group is cyclic of order $m$ and is generated by $\theta$. Consequently, any automorphism $\sigma$ can be expressed as $\sigma = \theta^e$ for some positive integer $e$. While skew polynomial rings over finite fields are both right and left Euclidean domain, this property does not extend to Galois rings. Nevertheless, a right Euclidean algorithm for skew polynomials over Galois rings is available, as stated below.
 \begin{proposition}\cite{McDonald}
     \label{Divisionalgo} Let $f(X), g(X) \in GR(p^s, p^{sm})[X; \sigma]$, and suppose the leading coefficient of $g(X)$ is a unit. Then, there exist unique polynomials $q(X), r(X) \in GR(p^s, p^{sm})[X; \sigma]$ such that \[ f(X) = q(X)\ast g(X) + r(X), \quad \text{where } r(X) = 0 \text{ or } \deg(r(X)) < \deg(g(X)). \] 
 \end{proposition}

 \noindent 
 \begin{remark}
      We say that $g(X)$ divides  $f(X)$ from right if $
f(X) = q(X) \ast g(X)
$
for some quotient $q(X)$. The definition of a left divisor follows symmetrically. It is important to note that these two concepts are equivalent if and only if the divisor $g(X)$ is an element from the center of the skew polynomial ring.

 \end{remark}
 In the skew polynomial ring $GR(p^s,p^{sm})[X;\sigma]$, evaluating skew polynomials is not the same as evaluating ordinary polynomials over commutative rings. Following the approach in \cite{Bhaintwal2012Skew} and the description given in \cite[pp.~15]{Jacobson1996}, we consider the following: let $\sigma$ be an automorphism of $GR(p^s,p^{sm})$, and let its extension to $GR(p^s,p^{sml})$ also be denoted by $\sigma$.
 For any $\beta \in GR(p^s,p^{sml})$, we have
	\[
	\mathcal{N}^{\sigma}_0(\beta)=1,\qquad
	\mathcal{N}^{\sigma}_{i+1}(\beta)=\sigma\bigl(\mathcal{N}^{\sigma}_i(\beta)\bigr)\beta,
	\]
	or, in expanded form,
	\[
	\mathcal{N}^{\sigma}_i(\beta)=\beta\cdot\sigma(\beta)\cdot\sigma^2(\beta)\cdots\sigma^{\,i-1}(\beta).
	\]
\begin{definition}\
	Let
	\(
	f(X)=\sum_{i=0}^{n-1} a_i X^i \in GR(p^s,p^{sm})[X;\sigma], \qquad a_i\in GR(p^s,p^{sm}),
	\)
	and \( \beta\in GR(p^s,p^{sm}) \).	The right evaluation of \( f \) at \( \beta \), denoted \( f(\beta) \), is defined as the remainder of the right division of \( f \) by \( X-\beta \) in \( GR(p^s,p^{sm})[X;\sigma] \). 	
\end{definition}
Equivalently, one has the closed form
	\(
	f(\beta)=\sum_{i=0}^{n-1} a_i\,\mathcal{N}^{\sigma}_i(\beta),
	\) we say that \( \beta \) is a right root of \( f \) if \( f(\beta)=0 \). Thus, the root condition can be written explicitly as
	\begin{eqnarray}\label{EQN1}
		\beta\ \text{is a right root of } f \quad \textup{if and only if}\quad
		\sum_{i=0}^{n-1} a_i\,\mathcal{N}^{\sigma}_i(\beta)=0.
	\end{eqnarray}




\section{Skew cyclic code and an MDS matrix}\label{Sec:Skew-cyclic-code}

In this section, we discuss $\sigma$-cyclic codes over a Galois ring, which form the main focus of this work. Since the principal results depend on such codes, we present their definition and fundamental properties. A linear code $C$ of length $n$ over $GR(p^s,p^{sm})$ is called $\sigma$-cyclic (or skew cyclic) if for every codeword  
\(
(c_0, c_1, \dots, c_{n-1}) \in C,
\)  
the $\sigma$-shifted vector  
\(
(\sigma(c_{n-1}), \sigma(c_0), \dots, \sigma(c_{n-2})) \in C.
\)  If the order of $\sigma$, denoted by $|\sigma|$, divides $n$, then $X^n-1$ belongs to the center $Z(GR(p^s,p^{sm})[X;\sigma])$. Consequently, $\langle X^n-1 \rangle$ is a two-sided ideal of the skew polynomial ring $GR(p^s,p^{sm})[X;\sigma]$, and hence we have   the quotient ring  
\[
\frac{GR(p^s,p^{sm})[X;\sigma]}{\langle X^n-1 \rangle}=\{a_0+a_1X+a_2X^2+\cdots+a_{n-1}X^{n-1}+\langle X^n-1 \rangle;~ a_i\in GR(p^s,p^{sm})\}.
\]  

Throughout this paper, we assume that $|\sigma|$ divides $n$. Moreover, we take $\gcd(n,p^s)=1$, and $g(X)$ to be a right divisor of $X^n-1$. Under these assumptions, the left ideal  
\(
^\ast\langle g(X)\rangle
\)  
represents a principal $\sigma$-cyclic code, and the following result describes its structure:

\begin{theorem}\cite[Theorem~2]{Bhaintwal2012Skew}\label{thm:principal-sigma-cyclic}
Let $C$ be a principal $\sigma$-cyclic code of length $n$ over $GR(p^s,p^{sm})$ generated by a monic polynomial $g(X) \in GR(p^s,p^{sm})[X;\sigma]$. Then, $C$ is a $GR(p^s,p^{sm})$-free module if and only if $g(X)$ is a right divisor of $X^n - 1$.
\end{theorem}

 Consequent to the above discussion, a $\sigma$-cyclic code admits two equivalent descriptions as follows:
\begin{itemize}
    \item[\textnormal{(i)}] It is a $GR(p^s,p^{sm})$-submodule of $GR(p^s,p^{sm})^n$ such that
    \[
    C := \{(c_0, c_1, \dots, c_{n-1}) \in GR(p^s,p^{sm})^n;~
    (\sigma(c_{n-1}), \sigma(c_0), \dots, \sigma(c_{n-2})) \in C\};
    \]
    \item[\textnormal{(ii)}] It is a left ideal in  the quotient ring $\frac{GR(p^s,p^{sm})[X;\sigma]}{\langle X^n-1\rangle}$, i.e.,
    \[
    C := ~^\ast\langle g(X)\rangle = \left\{\,a(X)g(X) \;\mid\; a(X)\in \frac{GR(p^s,p^{sm})[X;\sigma]}{\langle X^n-1 \rangle}\right\}.
    \]
\end{itemize}

\noindent In the design of linear diffusion layers for block ciphers and hash functions (see, e.g., \cite{Daemen2002the, Guo2011the, Guo2011TheLED}), the input and output of diffusion layers in those applications are generally of same size. Thus, we are interested only in square matrices that are MDS. 
In particular, our goal is to construct $[n=2k, k, d=k+1]$ codes for some positive integer $k$, from which we can obtain an MDS matrix of size $k \times k$ over $GR(p^s,p^{sm})$.

To this end, let 
\(
    g(X) = g_0 + g_1X + \cdots + g_{k-1}X^{k-1} + X^k \in GR(p^s,p^{sm})[X;\sigma]
\) 
be a right divisor of $X^{n} - 1$, and let $C$ be the principal $\sigma$-cyclic code, i.e., 
\(
    C = ^\ast\langle g(X) \rangle.
\)  A generator matrix of $C$ is given by
\begin{eqnarray}\label{Equation_gen0}
  	G = \begin{bmatrix}
  		g(X)\\
        X\ast g(X)\\
        \vdots\\
        X^{n-k-1}\ast g(X)
  	\end{bmatrix}_{n-k\times n}.
\end{eqnarray} 

In systematic form, this generator matrix can be rewritten as
\begin{eqnarray}\label{Equation_gen}
  	\begin{bmatrix}
  		- X^k \bmod_\ast g & 1 & 0 & \cdots & 0 \\
  		- X^{k+1} \bmod_\ast g & 0 & 1 & \cdots & 0 \\
  		\vdots & \vdots & \vdots & \ddots & \vdots \\
  		- X^{n-1} \bmod_\ast g & 0 & 0 & \cdots & 1
  	\end{bmatrix}_{n-k\times n}.
\end{eqnarray}

In particular, if we take $n=2k,$ (\ref{Equation_gen}) reduce to 

$$\begin{bmatrix}
  		- X^k \bmod_\ast g & 1 & 0 & \cdots & 0 \\
  		- X^{k+1} \bmod_\ast g & 0 & 1 & \cdots & 0 \\
  		\vdots & \vdots & \vdots & \ddots & \vdots \\
  		- X^{2k-1} \bmod_\ast g & 0 & 0 & \cdots & 1
  	\end{bmatrix}_{k\times 2k}=\begin{bmatrix}
  	    g_0&g_1&g_2&\dots&g_{k-1}&1&0&\dots&0\\
        0&g^{[1]}_0&g^{[1]}_1&\dots&g^{[1]}_{k-2}&g^{[1]}_{k-1}&1&\dots&0\\
        \vdots&\vdots&\vdots&\ddots&\vdots&\vdots\\
        0&0&0&\dots&g^{[k-1]}_0&g^{[k-1]}_1&\dots&g^{[k-1]}_{k-1}&1\\
  	\end{bmatrix}_{k\times 2k}.$$

Now, isolating the redundant part, we obtain the matrix

\begin{eqnarray} \label{eqn5.3-GR}
  	N_g = \begin{bmatrix}
  		X^k \bmod_\ast g \\
  		X^{k+1} \bmod_\ast g \\
  		\vdots \\
  		X^{2k-1} \bmod_\ast g
  	\end{bmatrix}_{k\times k}.
\end{eqnarray}

\noindent To analyze $N_g$, it is convenient to introduce the companion matrix associated with 
the polynomial $g(X)$.

\begin{definition}
Let 
\(
g(X) = g_0 + g_1 X + g_2 X^2 + \cdots + g_{k-1} X^{k-1} + X^k \in GR(p^s,p^{sm})[X;\sigma]
\)
be a monic polynomial of degree \(k\) over the ring \(GR(p^s,p^{sm})\).
The companion matrix of \(g(X)\), denoted by \(C_g\), is defined as
\[
C_g =
\begin{bmatrix}
0 & 1 & 0 & \cdots  & 0 \\
0 & 0 & 1 & \cdots  & 0 \\
\vdots & \vdots & \vdots & \ddots & \vdots \\
0 & 0 & 0 & \cdots  & 1 \\
-g_0 & -g_1 & -g_2 & \cdots  &- g_{k-1}
\end{bmatrix}_{k \times k}.
\]
\end{definition}

For each integer \(i \geq 0\), we denote by \(C_g^{[i]}\) the matrix obtained by 
applying \(\sigma^i\) to every entry of \(C_g\). Explicitly,
\[
C_g^{[i]} =
\begin{bmatrix}
0 & 1 & 0 & \cdots  & 0 \\
0 & 0 & 1 & \cdots & 0 \\
\vdots & \vdots & \vdots & \ddots &  \vdots \\
0 & 0 & 0 & \cdots &  1 \\
-\sigma^i(g_0) &- \sigma^i(g_1) &- \sigma^i(g_2) & \cdots &- \sigma^i(g_{k-1})
\end{bmatrix}_{k \times k}.
\]
Note that, if we have a square matrix $A$ of order $k$, say $A=(a_{ij})$, where $a_{ij}\in GR(p^s,p^{sm})$ and $1\leq i,~j \leq k$ then $A^{[t]}=(\sigma^t(a_{ij}))$ for some $t\in \mathbb{N}\cup \{0\}.$

Following the approach of Cauchois et al.~\cite[Theorem 2]{Cauchois2016direct}, the subsequent lemma provides a decomposition of the redundant matrix $N_g$ from (\ref{eqn5.3-GR}) into a product of successive \(\sigma\)-twists of the companion matrix.

\begin{lemma}\label{companiontocyclic}
    Let $g(X)=g_0+g_1X+\cdots+g_{k-1}X^{k-1}+X^k\in GR(p^s,p^{sm})[X;\sigma]$ be a right divisor of $X^n-1.$ 
    Then, for any positive integer $t$,
    \[
    C^{[t-1]}_g \cdot C^{[t-2]}_g \cdots C^{[1]}_g \cdot C_g =
    \begin{bmatrix}
        X^t \bmod_\ast g\\
        X^{t+1} \bmod_\ast g\\
        \vdots\\
        X^{t+k-1} \bmod_\ast g
    \end{bmatrix}_{k\times k}.
    \]
\end{lemma}
\begin{proof}
    	We show by strong induction that for all integers $i\geq 1$, the following property $P_i$ is satisfied:
	\begin{eqnarray}\label{eqn201}
		P_i&:&\begin{bmatrix}
			X^i ~ \mod _*g \\
			X^{i+1} ~ \mod _*g  \\
			\vdots\\
			X^{i+k-1} ~ \mod _*g  \\
		\end{bmatrix}_{k\times k}=C^{[i-1]}_g\cdot C^{[i-2]}_g\cdots C^{[1]}_g\cdot C_g.
	\end{eqnarray}	
    
    For $i=1,$ we have 
	\begin{eqnarray}
		P_0&:&\begin{bmatrix}
			X^1 ~ \mod _*g  \\
			X^{2} ~ \mod _*g  \\
			\vdots\\
			X^{k} ~ \mod _*g  \\
			
		\end{bmatrix}_{k\times k}=C_g.
	\end{eqnarray}
    
	Thus, (\ref{eqn201}) is true for $i=1$. Suppose our assumption is true for $1,~2,~\dots,~i-1$. Then, we have
	
	
	$$\begin{bmatrix}\label{eqn5.8}
		X^i ~ \mod _*g  \\
		X^{i+1} ~ \mod _*g  \\
		\vdots\\
		X^{i+k-1} ~ \mod _*g  \\
	\end{bmatrix}_{k\times k}=
	C_g^{[i-1]}\cdot C_g^{[i-2]}\cdots C_g^{[1]}\cdot C_g.$$
	Now, we prove (\ref{eqn201}) for $i$. For this, we compute 
    \begin{eqnarray*}
        \notag&~&C^{[i]}_g C_g^{[i-1]}\cdot C_g^{[i-2]}\cdots C_g^{[1]}\cdot C_g\\
		\notag&~~~~~&	=\begin{bmatrix}
			0&1&0&\cdots&0\\
			0&0&1&\cdots&0\\
			\vdots&\vdots&\vdots&\ddots&\vdots&\\
			0&0&0&\cdots&1\\
			\sigma^{i}(g_0)&\sigma^{i}(g_1)&\sigma^{i}(g_2)&\cdots&\sigma^{i}(g_{k-1})
		\end{bmatrix}_{k\times k} \cdot \begin{bmatrix}
			X^i ~ \mod _*g  \\
			X^{i+1} ~ \mod _*g \\
			\vdots\\
			X^{i+k-1} ~ \mod _*g \\
		\end{bmatrix}_{k\times k}\\
    \end{eqnarray*}
	\begin{eqnarray}\label{eqn11}
		\notag&=&\begin{bmatrix}
			X^{i+1}~ \mod _*g  \\
			X^{i+2}~ \mod _*g  \\
			\vdots\\
			X^{i+k-1}~\mod_*g\\
            \sigma^{i}(g_0)X^i+\cdots+\sigma^{i}(g_{k-1})X^{i+k-1}~ \mod _*g  
		\end{bmatrix}_{k\times k}\\
		&=&\begin{bmatrix}
			X^{i+1}~ \mod _*g \\
			X^{i+2}~ \mod _*g  \\
			\vdots\\
			\Big(\sum_{j=0}^{k-1}\sigma^{i}(g_j)X^{j+i}\Big)~ \mod _*g  \\
		\end{bmatrix}_{k\times k}.
	\end{eqnarray}
	\normalsize
    
	Since,
	\begin{eqnarray*}
		X^i*g_t&=&\sigma^i(g_t)X^i,
	\end{eqnarray*}
    
	we have
	\begin{eqnarray}\label{eqn12}
		X^{i+k}=X^i\ast X^k\notag&=& X^i*(g_0+g_1X+g_2X^2+\cdots+g_{k-1}X^{k-1})\\
\notag&=&\sigma^i(g_0)X^i+\sigma^i(g_1)X^{i+1}+\cdots+\sigma^i(g_{k-1})X^{i+k-1}\\
		&=&\sum_{j=0}^{k-1}\sigma^i (g_j)X^{i+j}.
	\end{eqnarray}
    
	From the last row of the matrix in (\ref{eqn11}), and (\ref{eqn12}), we obtain
	\begin{eqnarray}\label{eqn13}
		\Big(\sum_{j=0}^{k-1}\sigma^{i}(g_j)X^{j+i}\Big)~ \mod _*g  &=&X^{i+k}~ \mod _*g.
	\end{eqnarray}

Combining (\ref{eqn11}) and  (\ref{eqn13}), it follows that
	\begin{eqnarray}\label{eqn202}
		C_g^{[i]}\cdot 	C_g^{[i-1]}\cdot C_g^{[i-2]}\cdots C_g^{[1]}\cdot C_g=\begin{bmatrix}
			X^{i+1} ~ \mod _*g  \\
			X^{i+2} ~ \mod _*g  \\
			\vdots\\
			X^{i+k} ~ \mod _*g  \\
		\end{bmatrix}_{k\times k}.
	\end{eqnarray}
    
	Hence, $P_i$ is true for all integer $i\geq 1.$
\end{proof}

\noindent In particular, if we take $t=k$, in Lemma \ref{companiontocyclic}, we obtain
\begin{eqnarray}\label{eqn5.9-GR}
    N_g&=&\begin{bmatrix}
			X^{k} ~ \mod _*g  \\
			X^{k+1} ~ \mod _*g  \\
			\vdots\\
			X^{2k-1} ~ \mod _*g  \\
		\end{bmatrix}_{k\times k}=C^{[k-1]}_g\cdot C^{[k-1]}_g\cdots C^{[1]}_g\cdot C_g.\end{eqnarray}

Following the work of Cauchois et al.~\cite{Cauchois2016direct}, we recall the definitions of recursive and quasi-recursive MDS matrices.

\begin{definition} 
A square matrix $M$ of order $k$ is said to be: 
\begin{enumerate} 
    \item[\textnormal{(i)}] \textbf{recursive MDS} if  $M^r$ is an MDS matrix for some positive integer $r$.
    \item[\textnormal{(ii)}] \textbf{quasi-recursive MDS} if the matrix product $M^{[r-1]} M^{{[r-2]}} \cdots M^{[1]} M$ is an MDS matrix for some positive integer $r$.
\end{enumerate}  
\end{definition}

\section{Characterization of polynomials that yield quasi-recursive MDS matrices}\label{Sec:polynomial-charac}

As $GR(p^s, p^{sm})$ is a finite commutative ring with unity, we shall denote $\mathcal{R} = GR(p^s, p^{sm})$ throughout this section. The focus of this section is the construction of quasi-recursive MDS matrices. We say that a monic polynomial $g(X) \in \mathcal{R}[X;\sigma]$ of degree $k$ yields a  quasi-recursive MDS matrix if there exists an integer $t \geq k$ for which 
the matrix
\(
M = C_g^{[t-1]} C_g^{[t-2]} \cdots C_g
\)
is  MDS. Since the case for matrices of size $1 \times 1$ is trivial, 
we henceforth assume that any polynomial considered for obtaining quasi 
recursive MDS matrices has degree $k = \deg(g) \geq 2$.

The next result provides a criterion for characterizing such polynomials, extending the work of \cite[Theorem~1]{Guptadirect} to the framework of skew polynomial rings. Subsequently, we present a more general form of this result. Furthermore, we establish key results, in particular Theorem~\ref{Theroem 5}, which play an important role in the direct construction of quasi-recursive MDS matrices over Galois rings in the later section.


\begin{theorem}\label{Weightthm1} Let  \(g(X)=g_0+g_1X+\cdots+g_{k-1}X^{k-1}+X^k\in \mathcal{R}[X;\sigma]\) be a right divisor of $X^{n}-1$, where $n\geq 2k$. Then,   $N_g$ is an MDS matrix if and only if the Hamming weight of any non-zero left multiple of $g(X)$ of degree less than $2k$ is at least $k+1$. \end{theorem}

\begin{proof} 
Consider \(g(X)=g_0+g_1X+\cdots+g_{k-1}X^{k-1}+X^k\in \mathcal{R}[X;\sigma]\) be a right divisor of $X^{n}-1$, where $n\geq 2k$. In view of (\ref{eqn5.9-GR}), we have

	$$N_g=C^{[k-1]}_g\cdot C^{[k-2]}_g\cdots C^{[1]}\cdot C_g=\begin{bmatrix}
		X^k~ \mod _*g  \\
		X^{k+1}~\ \mod _*g  \\
		\vdots\\
		X^{2k-1}~\ \mod _*g  
	\end{bmatrix}_{k\times k}.$$

    Let $C$ be the code with generator matrix
    $$\bar{G}_{k \times 2k}=\begin{bmatrix}
		g(X)\\
        X\ast g(X)\\
        \vdots\\
        X^{k-1}\ast g(X)
	\end{bmatrix}_{k\times 2k}.$$
    
    Note that the rows of $\bar{G}_{k\times 2k}$ are linearly independent over \(\mathcal{R}\). The elements of \(C\) correspond to the left multiples of \(g(X)\) of degree less than \(2k\). On account of (\ref{Equation_gen0}) and (\ref{Equation_gen}), one can see that matrix  $\big[\, C_g^{[k-1]} C_g^{[k-2]} \cdots C_g \,\big|\, I_k \big]$, also a generator matrix of the code $C$. Hence the proof. 
	   \end{proof}

It follows from Theorem~\ref{Weightthm1}, that the resulting MDS matrix is  quasi involutory if $n=2k$. This holds for any choice of automorphism $\sigma$, as long as the order of $\sigma$ divides $2k$. The following proposition makes this precise.

 \begin{proposition}\label{quasiinvolutory} Let $\sigma$ be any automorphism on $\mathcal{R}$, and let $g(X) \in \mathcal{R}[X; \sigma]$ be a monic polynomial of degree $k$ that is a right divisor of $X^{2k}-1$. Then, the matrix $N_g = C_g^{[k-1]} C_g^{[k-2]} \cdots C_g$ satisfies the relation $N_g^{[k]}N_g = I_k$.
  \end{proposition}

  \begin{proof} 
  	Consider $N_g = C_g^{[k-1]} C_g^{[k-2]} \cdots C_g$. Now, compute \[ N_g^{[k]} = (C_g^{[k-1]} C_g^{[k-2]} \cdots C_g)^{[k]} = C_g^{[2k-1]} C_g^{[2k-2]} \cdots C_g^{[k]}. \] 
    
    The product $N_g^{[k]}N_g$ is therefore given by 
    \begin{eqnarray*}
          N_g^{[k]}N_g &=& (C_g^{[2k-1]} C_g^{[2k-2]} \cdots C_g^{[k]}) \cdot (C_g^{[k-1]} C_g^{[k-2]} \cdots C_g)\\
          &=& C_g^{[2k-1]} C_g^{[2k-2]}   \cdots C_g. 
    \end{eqnarray*}
    
    Invoking of Lemma \ref{companiontocyclic} and the condition $X^{2k}-1=0 \mod_\ast g$, we obtain
    \begin{eqnarray*}
        N_g^{[k]}N_g&=&C_g^{[2k-1]} C_g^{[2k-2]}   \cdots C_g.\\
        &=&\begin{bmatrix}
		X^{2k}~ \mod _*g  \\
		X^{2k+1}~\ \mod _*g  \\
		\vdots\\
		X^{3k-1}~\ \mod _*g  
	\end{bmatrix}_{k\times k}\\
    &=&\begin{bmatrix}
		X^{0}~ \mod _*g  \\
		X^{1}~\ \mod _*g  \\
		\vdots\\
		X^{k-1}~\ \mod _*g  
	\end{bmatrix}_{k\times k}\\
    &=&I_k.
    \end{eqnarray*}
    \end{proof}
  		
The matrix obtained in Proposition \ref{quasiinvolutory} is referred to as quasi involutory matrix.

\begin{remark} 
    Consider the Galois ring $ GR(p^s, p^{2sk})$. If we consider the automorphism $\sigma: GR(p^s, p^{2sk}) \rightarrow GR(p^s, p^{2sk})$ defined by 
    \[ \sigma\left(\sum_{i=0}^{2k-1} a_i \zeta^i\right) = \sum_{i=0}^{2k-1} a_i \zeta^{i p^2}, \] 
    where $\{a_i;~ 0\leq i \leq 2k-1\} \subset \mathbb{Z}_{p^s}$ and $\zeta$ is a primitive root of a  basic irreducible polynomial  of degree $2k$ over $\mathbb{Z}_{p^s}$, then for any monic polynomial $g(X)$ that is a right divisor of $X^{2k}-1$, the product $N_g= C^{[k-1]}_g\cdot C^{[k-2]}_g\cdots C^{[1]}_g\cdot C_g$ will be an involutory matrix. An explicit example illustrating this fact is given below.
\end{remark}

\begin{example}
Let $GR(5^2, 5^6)=\frac{\mathbb{Z}_{5^2}[Y]}{\langle 3+3Y+Y^3\rangle}=\{a_0+a_1\zeta+a_2\zeta^2;~\zeta=Y+\langle 3+3Y+Y^3\rangle, ~a_0,a_1,a_2\in \mathbb{Z}_{5^2} \}$, and define $\sigma :GR(5^2, 5^6) \longrightarrow GR(5^2, 5^6)$ by 
\(
\sigma(a_0+a_1\zeta+a_2\zeta^2) = a_0+a_1\zeta^{25}+a_2\zeta^{50}.
\) Since $X^3 + 2X^2+2X+1 \in GR(5^2,5^6)[X]$ is a right divisor of $X^6 - 1$, take
\(
g(X) = X^3 + 2X^2+2X+1.
\)
Then, the companion matrix of $g(X)$ is
\[
C_g = \begin{bmatrix}
0 & 1&0 \\
0&0&1\\
24 &23&23 
\end{bmatrix}_{3\times 3}, 
\qquad 
N_g=C_g^{[2]}\cdot C_g^{[1]}\cdot C_g = \begin{bmatrix}
24&23&23\\
2&3&2\\
23&23&24
\end{bmatrix}_{3\times 3}.
\]
It can be verified that the product 
\(N_g\cdot N_g 
= I_{3}
\) and $N_g$ is an MDS matrix. Hence, $N_g$ is an involutory MDS matrix of order 3.
\end{example}

\noindent Let $C$ be a $\sigma$-cyclic code of length $n$ over the ring $\mathcal{R}$, generated by a monic polynomial 
$g(X) \in \mathcal{R}[X; \sigma]$ which is a right divisor of $X^n - 1$. A shortened skew cyclic code can be 
constructed from $C$ by selecting only the codewords that have zero entries in a specific set of coordinate positions, 
and then deleting these positions from all selected codewords. In this context, we consider a specific type of shortened 
code of length $2k$ derived from the parent code $C$, where $n\geq 2k $. We are interested in a code whose systematic generator matrix is 
of the form $[I_k \mid M]$, where the information symbols correspond to the first $k$ positions 
$\{0, 1, \dots, k-1\}$, and the parity symbols correspond to $k$ positions starting from an offset $t$, 
$\{t, t+1, \dots, t+k-1\}$, where $t\geq k$. In this setting, the support of a polynomial 
\(
f(X) = \sum_{i=0}^n a_i X^i \in \mathcal{R}[X; \sigma],
\)
denoted by $\operatorname{supp}(f)$ (or sometimes $\mathrm{Supp}(f)$), is defined as the set of indices $i$ for which 
the coefficient $a_i$ is nonzero; that is,
\[
\operatorname{supp}(f) = \{\, i \in \mathbb{N}\cup \{0\} \mid a_i \neq 0 \text{ in } \mathcal{R} \,\}.
\]

The codewords of the shortened code under consideration can thus be identified with polynomials whose support is 
contained within the set of exponents
\[
E = \{0, 1, \dots, k-1\} \cup \{t, t+1, \dots, t+k-1\}.
\]

Following the proof of~\cite[Theorem~2]{gupta2017towards}, we establish the corresponding result in the setting of the skew polynomial ring.

\begin{theorem}\label{equivalent condition}  Let $g(X)=g_0+g_1X+\cdots+g_{k-1}X^{k-1}+X^k \in \mathcal{R}[X; \sigma]$ be a  right divisor of $X^n - 1$. For an integer $t$ such that $k \leq t \leq n-k$, define the set of exponents $E = \{0, 1, \dots, k-1\} \cup \{t, t+1, \dots, t+k-1\}$. Then, the matrix $M = C_g^{[t-1]} C_g^{[t-2]} \cdots C_g$ is an MDS matrix if and only if the Hamming weight of any non-zero polynomial of the form $f(X) = c(X)\ast g(X)$ with support contained in $E$ is greater than $k$. 
\end{theorem}
\begin{proof}Consider \(g(X)=g_0+g_1X+\cdots+g_{k-1}X^{k-1}+X^k\in \mathcal{R}[X;\sigma]\) be a right divisor of $X^{n}-1$. Using $i=t$ in Lemma \ref{companiontocyclic}, we obtain
	$$C^{[t-1]}_g\cdot C^{[t-2]}_g\cdots C^{[1]}\cdot C_g=\begin{bmatrix}
		X^t~ \mod _*g  \\
		X^{t+1}~\ \mod _*g  \\
		\vdots\\
		X^{t+k-1}~\ \mod _*g  
	\end{bmatrix}_{k \times k}.$$
    
	By the Corollary \ref{Fact1}, the matrix $N_g$ is MDS if and only if any $k$ rows of matrix $\bar{G}=\left[ \frac{I_k}{N_g}\right]$ are linearly independent, where
	$$\bar{G}=\Big[\frac{I_{k}}{C_g^{[t-1]} C_g^{[t-2]} \cdots C_g}\Big]=\begin{bmatrix}
		1\\X\\X^2\\
		\dots\\
		X^{k-1}\\
		X^t~ \mod _*g  \\
		X^{t+1} \mod _*g  \\
		\dots\\
		X^{t+k-1} \mod _*g  
	\end{bmatrix}_{2k\times k}.$$
    
	Now, it is easy to see that the latter condition is equivalent to requiring that the weight of any nonzero left multiple of \( g(X) \) of the form \( \sum_{e \in E} f_e X^e \) is at least \( k \). Hence, this proves the theorem.
 \end{proof}
 
\begin{example}Consider the Galois ring 
\( GR(2^2, 2^8)=\frac{\mathbb{Z}_4[Y]}{\langle Y^4+Y+1 \rangle} \) endowed with Frobenius automorphism 
$ \sigma : GR(2^2, 2^8) \rightarrow GR(2^2, 2^8).$ If we consider the polynomial
\(
g(X) = X^4 + (2\alpha^3 + \alpha^2 + 3\alpha + 3)X^3 + \alpha^3 X^2 +X+ 1
\in GR(2^2, 2^8)[X; \sigma],
\) where $\alpha=Y+\langle Y^4+Y+1 \rangle$, 
then, it can be verified that the matrix
\[
M = C_g^{[3]} C_g^{[2]} C_g^{[1]} C_g
\]
is an MDS matrix of order \( 4 \).
\end{example}

\begin{proposition}\label{Prop3} Let   $g(X) = g_0 + g_1X + \cdots + g_{k-1}X^{k-1} + X^k\in \mathcal{R}[X;\sigma]$ be a right divisor of  $X^n-1$. If $g_0$ is a non-unit in $\mathcal{R}$, then for any positive integer $s$, the matrix $M = C_g^{[s-1]} C_g^{[s-2]} \cdots C_g$ is non  MDS. \end{proposition}
\begin{proof}
Consider $g(X)=g_0+g_1X+g_2X^2+\cdots+g_{k}+X^k\in \mathcal{R}[X;\sigma]$ such that $g_0$ is non-unit. This implies that $\det(C_g)\in \mathcal{N}(\mathcal{R})$. Hence, the $C_g^{[s-1]} C_g^{[s-2]} \cdots C_g$ is non MDS.
\end{proof}
Next, we turn to the construction of quasi-recursive MDS matrices using the right roots of the polynomial $g(X)=g_0+g_1X+g_2X^2+\cdots+g_{k-1}X^{k-1}+X^k$. Vandermonde matrices play an important role in this context, as they provide a natural framework for analyzing polynomial roots and understanding how MDS matrices can be obtained. Motivated by this relation, we now recall the definition of Vandermonde matrices in the setting of skew polynomial rings, 
particularly over the skew polynomial ring
\(
GR(p^s, p^{sm})[X;\sigma],
\) 
as introduced by Boulagouaz and Leroy~\cite[Definition~4]{Boulagouaz2013codes}.

\begin{definition}\label{Def01}
	\begin{itemize}
		\item[\textnormal{(i)}]  A monic polynomial \( g(X) \in  \mathcal{R}[X; \sigma] \) of degree \( k \) is a Wedderburn polnomial 
		if there exist elements \( a_1, \dots, a_k \in \mathcal{R} \) such that  
		\[
		\mathcal{R} g(X) = \bigcap_{i=1}^k \mathcal{R}(X - a_i).
		\]  
		These polynomials will be termed as W-polynomials.
		
		\item[\textnormal{(ii)}] The \( k \times n \) generalized Vandermonde matrix defined by \( a_1, \dots, a_k \) is given by
		$$
		V^\sigma_n(a_1, \dots, a_k) =\begin{bmatrix}
1 & a_1 & \mathcal{N}^\sigma_2(a_1) & \cdots & \mathcal{N}^\sigma_{n-1}(a_1) \\
1 & a_2 & \mathcal{N}^\sigma_2(a_2) & \cdots & \mathcal{N}^\sigma_{n-1}(a_2) \\
\vdots & \vdots & \vdots & \ddots & \vdots \\
1 & a_k & \mathcal{N}^\sigma_2(a_k) & \cdots & \mathcal{N}^\sigma_{n-1}(a_k)
\end{bmatrix}_{k\times n}.
$$
	\end{itemize}
	
\end{definition}

Let $V^\sigma_n(a_1, a_2, \dots, a_k)$ denote the generalized Vandermonde matrix.  
We write $V^\sigma_n(a_1, a_2, \dots, a_k; Z)$ for the submatrix obtained by selecting the columns indexed by $Z = \{i_1, i_2, \dots, i_k\}$.  
For the special case \( n = k \) and the identity automorphism \( \sigma = \mathrm{id} \) on \( \mathcal{R} \),  
the generalized Vandermonde matrix \( V^\sigma_n(a_1, a_2, \dots, a_k) \) reduces to the classical Vandermonde matrix defined below.

\begin{definition}
Let \( a_1,a_2,\dots,a_k \in GR(p^{s}, p^{sml}) \). Then, the classical Vandermonde matrix over \( \mathcal{R} \) is defined as
\[
V^{\mathrm{id}}_k(a_1, a_2, \dots, a_k) =
\begin{bmatrix}
1 & a_1 & a_1^2 & \cdots & a_1^{k-1} \\
1 & a_2 & a_2^2 & \cdots & a_2^{k-1} \\
\vdots & \vdots & \vdots & \ddots & \vdots \\
1 & a_k & a_k^2 & \cdots & a_k^{k-1}
\end{bmatrix}_{k\times k},
\]

and its determinant can be expressed as follows:
\[\det\bigl(V^{\mathrm{id}}_k(a_1, a_2, \dots, a_k))\bigr) = \prod_{1\leq i < j \leq k} (a_j - a_i).
	\]
\end{definition}

To further characterize the generalized Vandermonde matrix \( V^\sigma_n(a_1, a_2, \dots, a_k) \), we consider an index set 
\( T = \{t_{1}, t_{2}, \dots, t_{k}\} \subset \mathbb{Z} \) with 
\( 0 \leq t_{1} < t_{2} < \cdots < t_{k} \).

\begin{definition}
Let \( a_1,a_2,\dots,a_k \in GR(p^{s}, p^{sml}) \) and 
\( T = \{t_{1}, t_{2}, \ldots, t_{k}\} \subset \mathbb{Z} \) with 
\( 0 \leq t_{1} < t_{2} < \cdots < t_{k} \). 
Then, the matrix
\[
V^{\mathrm{id}}_n(a_1, a_2, \dots, a_k;~ T) =
\begin{bmatrix}
a_{1}^{t_{1}} & a_{1}^{t_{2}} & \cdots & a_{1}^{t_{k}} \\
a_{2}^{t_{1}} & a_{2}^{t_{2}} & \cdots & a_{2}^{t_{k}} \\
\vdots & \vdots & \ddots & \vdots \\
a_{k}^{t_{1}} & a_{k}^{t_{2}} & \cdots & a_{k}^{t_{k}}
\end{bmatrix}_{k \times k}
\]

is called a generalized Vandermonde matrix with respect to \( T \).
\end{definition}

\begin{remark}
We observe that if we take $T=\{0,1,2,\dots,k-1\}$, then the matrix reduces to a  classical Vandermonde matrix $V^{\mathrm{id}}_k(a_1, a_2, \dots, a_k)$.
\end{remark}

As discussed in \cite{GuptanearMDS}, the authors show that the determinants corresponding to the above three cases can be computed in terms of the determinant of the Vandermonde matrix.

\begin{corollary}\cite[Corollary 3]{GuptanearMDS}\label{Corol3}
Let $T = \{0, 1, 2, \ldots, k-2, k\}$. Then, 
\[
\det\!\left(V^{\mathrm{id}}_n(a_1, a_2, \dots, a_k;~ T)\right) 
= \det\!\left(V^{\mathrm{id}}_k(a_1, a_2, \dots, a_k)\right) \left(\sum_{i=1}^{k} a_{i}\right).
\]
\end{corollary}
\begin{corollary}\cite[Corollary 5]{GuptanearMDS}\label{Corol5}
Let $T = \{0, 2, 3, \ldots, k-1, k+1\}.$
Then,  the determinant of $V^{\mathrm{id}}_n(a_1, a_2, \dots, a_k; T)$ is  given by 
\[
\det\!\left(V^{\mathrm{id}}_n(a_1, a_2, \dots, a_k;~ T_3)\right) 
= \det\!\left(V^{\mathrm{id}}_k(a_1, a_2, \dots, a_k)\right) 
\left(\prod_{i=1}^{k} a_{i}\right)
\Bigg[\sum_{i=1}^{k}a_i \sum_{i=1}^{k}a^{-1}_i-1\Bigg].
\]
\end{corollary}

%
%
Motivated by the relationship between polynomial roots and generalized Vandermonde matrices, we adapt the result of Boulagouaz and Leroy~\cite[Proposition~4]{Boulagouaz2013codes} to the framework of $\sigma$-cyclic codes over Galois rings.

\begin{proposition}\label{Pro005}
Let   $ g(X) \in \mathcal{R}[X;\sigma]$ be monic polynomials of degree $k$. Suppose that $g(X)$ is a $W$-polynomial with $X^n-1 \in ~^\ast\langle g(X)\rangle$, and  $C$ be the $\sigma$- cyclic code of length $n$ corresponding 
to the left cyclic $\mathcal{R}$-module $\mathcal{R} g(X)/\mathcal{R} (X^n-1)$. Let $a_1, a_2, \dots, a_k \in \mathcal{R}$ be such that 
\[
\mathcal{R} g(X) = \bigcap_{i=1}^k \mathcal{R}(X-a_i).
\] 

Then, for a codeword $(c_0, c_1, \dots, c_{n-1}) \in \mathcal{R}^n$, we have 
\[
(c_0, c_1, \dots, c_{n-1}) \in \mathcal{C} 
\quad \Longleftrightarrow \quad 
(c_0, c_1, \dots, c_{n-1}) V_n^\sigma(a_1, a_2, \dots, a_k)^T= \textbf{0}.
\]
\end{proposition}

 Now, we give an analogous result of the above proposition  in the framework of our Theorem \ref{equivalent condition}.
 \begin{theorem}
	Let  $g( X )\in \mathcal{R}[X;\sigma]$ be W-polynomial of degree $k$ such that 	\(
\mathcal{R} g(X) = \bigcap_{i=1}^k \mathcal{R}(X - a_i)
\), where  $a_1,a_2,\dots ,a_{k}$ are distinct elements, and let \( C \) be the code generated by $[-C^{[t-1]}_g\cdots C^{[1]}_g\cdot C_g|I_k]$ for some $t$, $k\leq t\leq n-k.$ Then, a vector $(f_0,f_1,\dots,f_{k-1},f_t,\dots,f_{t+k-1})\in \mathcal{R}^{2k}$ is an element of $C$ iff $$[f_0,f_1,\dots,f_{k-1},f_t,\dots,f_{t+k-1}]{\begin{bmatrix}
1&\mathcal{N}^\sigma_1(a_1)&\cdots&\mathcal{N}^\sigma_{k-1}(a_1)&\mathcal{N}^\sigma_t(a_1)&\cdots&\mathcal{N}^\sigma_{t+k-1}(a_1)\\
1&\mathcal{N}^\sigma_2(a_1)&\cdots&\mathcal{N}^\sigma_{k-1}(a_2)&\mathcal{N}^\sigma_t(a_2)&\cdots&\mathcal{N}^\sigma_{t+k-1}(a_2)\\
\vdots&\vdots&\ddots&\vdots&\vdots&\ddots&\vdots\\
1&\mathcal{N}^\sigma_2(a_k)&\cdots&\mathcal{N}^\sigma_{k-1}(a_k)&\mathcal{N}^\sigma_t(a_k)&\cdots&\mathcal{N}^\sigma_{t+k-1}(a_k)\\
\end{bmatrix}^T_{k\times 2k}}=\textbf{0}.$$  
 \end{theorem}
\begin{remark}
    We use the notation$$V^{\sigma}_n(a_1,a_2,\dots,a_k;E)=\begin{bmatrix}
1&\mathcal{N}^\sigma_1(a_1)&\cdots&\mathcal{N}^\sigma_{k-1}(a_1)&\mathcal{N}^\sigma_t(a_1)&\cdots&\mathcal{N}^\sigma_{t+k-1}(a_1)\\
1&\mathcal{N}^\sigma_2(a_1)&\cdots&\mathcal{N}^\sigma_{k-1}(a_2)&\mathcal{N}^\sigma_t(a_2)&\cdots&\mathcal{N}^\sigma_{t+k-1}(a_2)\\
\vdots&\vdots&\ddots&\vdots&\vdots&\ddots&\vdots\\
1&\mathcal{N}^\sigma_2(a_k)&\cdots&\mathcal{N}^\sigma_{k-1}(a_k)&\mathcal{N}^\sigma_t(a_k)&\cdots&\mathcal{N}^\sigma_{t+k-1}(a_k)\\
\end{bmatrix}_{k\times 2k},$$\\ where $E=\{0,1,2,\dots,k-1,t,t+1,\dots,t+k-1\}.$
\end{remark}
The preceding results allow us to state the following theorem:
\begin{theorem}\label{Theroem 5}
	Let $g( X )\in \mathcal{R}[X;\sigma]$ be W-polynomial of degree $k$ such that 	\(
	\mathcal{R} g(X) = \bigcap_{i=1}^k \mathcal{R}(X - a_i)
	\), where  $a_1,a_2,\dots ,a_{k}$ are distinct elements. Let $t$ be an integer with $k\leq t \leq n-k$. Then, the matrix $C^{[t-1]}_g\cdots C^{[1]}_g\cdot C_g$ is MDS if and only if any $k$ columns of the matrix $ V^{\sigma}_n(a_1,a_2,\dots,a_k;E)$  are linearly independent over $\mathcal{R}.$ 
\end{theorem}
\begin{proof}
Assume first that the matrix 
\(
M = C_g^{[t-1]} \cdots C^{[1]}_g \cdot C_g
\) 
is MDS. According to Theorem~\ref{equivalent condition}, this condition is equivalent to every nonzero left multiple of $g(X)$ with support in 
\(
E = \{0, \dots, k-1,~t, \dots, t+k-1\}
\) 
has Hamming weight greater than $k$.   Suppose there exists a set of $k$ columns of $V^\sigma_n(a_1, \dots, a_k)$ indexed by $E' = \{i_1, \dots, i_k\} \subset E$ such that  
\[
\sum_{j=1}^{k} c_{i_j} 
\begin{pmatrix}
\mathcal{N}^\sigma_{i_j}(a_1) \\ 
\mathcal{N}^\sigma_{i_j}(a_2) \\ 
\vdots \\ 
\mathcal{N}^\sigma_{i_j}(a_k)
\end{pmatrix} 
= \mathbf{0},
\]

for some $c_{i_1}, \dots, c_{i_k} \in \mathcal{R}$.  
This means that for each $s \in \{1, \dots, k\}$, we have 
\[
\sum_{j=1}^{k} c_{i_j} \mathcal{N}^\sigma_{i_j}(a_s) = 0.
\]

Consequently, in view of (\ref{EQN1}), every $a_s$ with $1 \leq s \leq k$ is a right root of
\(
f(X) = \sum_{j=1}^{k} c_{i_j} X^{i_j},
\)
so $f(X) \in \mathcal{R} g(X) = \bigcap_{i=1}^k \mathcal{R}(X - a_i)$. Since $E' \subseteq E$ and $\mathrm{wt}(f) \leq k$, the MDS property of $M$ implies that $f(X)$ must be the zero polynomial. It follows immediately that $c_{i_j} = 0$ for each $j$, $1\leq j\leq k$. Thus, any $k$ columns of $V^\sigma_n(a_1, \dots, a_k;E)$ are linearly independent.  Conversely,  let  
\(
0\neq f(X) = \sum_{e \in E} c_e X^e \in \mathcal{R} g(X).
\)  
Then, we have $f(X) \in \bigcap_{i=1}^k \mathcal{R}(X-a_i)$.  This implies that 
\[
\sum_{e\in E} c_{e} \mathcal{N}^\sigma_{e}(a_s) = 0 
\quad \text{for all } s \in \{1, \dots, k\}.
\]

Since every $k$-column of $V^\sigma_n(a_1,a_2,\dots,a_k;E)$ are linearly independent, so $\mathrm{wt}(f(X))>k.$ This complete the proof.
\end{proof}
\begin{corollary}\label{Corollarymain}
    Let  $g( X )\in \mathcal{R}[X;\sigma]$ be W-polynomial of degree $k$ such that 	\(
	\mathcal{R} g(X) = \bigcap_{i=1}^k \mathcal{R}(X - a_i)
	\), where  $a_1,a_2,\dots ,a_{k}$ are distinct elements. Let $t$ be an integer with $k\leq t \leq n-k$. Then, the matrix $C^{[t-1]}_g\cdots C^{[1]}_g\cdot C_g$ is MDS if and only if the determinant of any $k\times k$ submatrix of the matrix $ V^{\sigma}_n(a_1,a_2,\dots,a_k;E)$  are in $\mathcal{U}(\mathcal{R}).$
\end{corollary}

\section{Study of quasi-recursive MDS matrices on the basis of simple roots}\label{Sec:MDS-simple-roots}

In this section, we introduce techniques for the direct construction of quasi-recursive MDS matrices, extending the existing approaches to the more general setting of Galois rings. Our approach is based on studying the right roots of $X^n-1$ in the extension of Galois ring, which requires recalling the link between skew polynomials and linearized polynomials, as outlined by Bhaintwal~\cite{Bhaintwal2012Skew}. In particular, the roots of $X^n - 1 \in GR(p^s,p^{sm})[X;\sigma]$ are essential for the construction of quasi-recursive MDS matrices. For this analysis, Bhaintwal~\cite{Bhaintwal2012Skew} extended the concept of linearized polynomials, originally introduced by Ore~\cite{Orespecialclass}, to the framework of Galois rings. Following his notation, we associate  $X^n-1$ to the linearized polynomial 
\[
    L(Y) = \sigma^n(Y) - Y.
\]

For $\sigma = \theta^e$ (where 
$e$ denotes a positive integer), reduction modulo $p$ gives 
\[
    \overline{L}(Y) = Y^{q_0^n} - Y,
\]
with $q_0 = p^e$, which is the classical $q_0$-linearized polynomial over $\mathbb{F}_{p^m}$.
 Since its derivative $\overline{L}'(Y) = -1$, the polynomial $\overline{L}(Y)$ is separable and thus has distinct roots in $\mathbb{F}_{q_0^n}$. Let the smallest field containing $\mathbb{F}_{p^m}$ and all the roots of $\overline{L(Y)}$ is $\mathbb{F}_{p^{ml}}$, where $l$ is a positive integer. Hence, $GR(p^s, p^{msl})$ is the splitting ring of $L(Y)$ over $GR(p^s, p^{sm})$.
 If $\xi$ is a primitive element of $GR(p^s,p^{sne})$, then the roots of $L(Y)$ are exactly the Teichm$\ddot{u}$ller elements of $GR(p^s,p^{sne})$, which are all distinct units. Moreover, as shown by Chaussade et al.~\cite[Lemma~4]{Chaussade2009}, $\sigma(\beta)/\beta$ is a right root of $X^n - 1$ precisely when $\beta$ is a root of the linearized polynomial $Y^{q_0^n} - Y$.
This reveals that $X^n-1$ have $q^n_0-2$ distinct nonzero roots in extension ring. Letting $\sigma = \theta^e$ and retaining the above notations for $GR(p^s, p^{msl})$ and $GR(p^s, p^{sne})$, a free $\sigma$-cyclic code $\mathcal{C}$ of length $n$ over $GR(p^s, p^{sml})$ can be represented by its generator polynomial 
\(
g(X) = g_0 + g_1X + \cdots + g_{k-1}X^{k-1} + X^k,
\)
which possesses elements $\xi^{l_j} \in GR(p^s, p^{sne})$, $1 \leq j \leq k$, with $\{l_1, l_2, \dots, l_k\} \subset \{0, 1, 2, \dots, q_0^n - 2\}$ as right roots. Equivalently, each $X - \xi^{l_j}$ is a right divisor of $g(X)$, where $\xi$ denotes a primitive element of $GR(p^s, p^{sne})$. The condition $(n, q) = 1$ ensures that $X^n - 1$, and consequently $g(X)$, has only simple right roots.  

Throughout this section, 
let  
\(
E = \{0, 1, 2, \dots, k-1, t, t+1, \dots, k+t-1\},
\)
and let $E' \subset E$ be defined as  
\(
E' = \{i_1, i_2, \dots, i_k\}, \quad \text{where } i_1 < i_2 < \cdots < i_k.
\)
We begin this section with the following result.

\begin{theorem}\label{Theorem 7}   
Let $a_j = \xi^{\,j + b - 1}$ for $j = 1, \dots, k$, where $b$ is an integer. 
Suppose that $a_1, \dots, a_k$ are the right roots of $g(X)$. Then, for an integer 
$t \geq k$, the matrix 
\[
M = C_g^{[t-1]} C_g^{[t-2]} \cdots C_g
\] 

is an \textup{MDS} matrix.
\end{theorem}

 \begin{proof} 
	As established in Corollary~\ref{Corollarymain}, the matrix $M$ is \textup{MDS} precisely when every $k \times k$ submatrix of the generalized Vandermonde matrix $V^\sigma_n(a_1,a_2,\dots,a_k;E)$ is nonsingular. To analyze this condition, let us focus on a particular submatrix $V$, obtained by selecting the columns indexed by $E' = \{i_1,i_2,\dots,i_k\} \subset E$. We then have

    \begin{align}\label{EQN12}
        V
        &=V^\sigma_n(a_1,a_2,\dots,a_k;E') 
		= \begin{pmatrix} 
			\mathcal{N}^\sigma_{i_1}(a_1) & \mathcal{N}^\sigma_{i_2}(a_1) & \cdots & \mathcal{N}^\sigma_{i_k}(a_1) \\ 
			\mathcal{N}^\sigma_{i_1}(a_2) & \mathcal{N}^\sigma_{i_2}(a_2) & \cdots & \mathcal{N}^\sigma_{i_k}(a_2) \\ 
			\vdots       & \vdots       & \ddots & \vdots       \\ 
			\mathcal{N}^\sigma_{i_1}(a_k) & \mathcal{N}^\sigma_{i_2}(a_k) & \cdots & \mathcal{N}^\sigma_{i_k}(a_k) 
		\end{pmatrix}_{k\times k}
    \end{align}


	\begin{align*}
		\notag&= \begin{pmatrix} 
			\mathcal{N}^\sigma_{i_1}(\xi^{b}) & \mathcal{N}^\sigma_{i_2}(\xi^{b}) & \cdots & \mathcal{N}^\sigma_{i_k}(\xi^{b}) \\ 
			\mathcal{N}^\sigma_{i_1}(\xi^{b+1}) & \mathcal{N}^\sigma_{i_2}(\xi^{b+1}) & \cdots & \mathcal{N}^\sigma_{i_k}(\xi^{b+1}) \\ 
			\vdots & \vdots & \ddots & \vdots \\ 
			\mathcal{N}^\sigma_{i_1}(\xi^{b+k-1}) & \mathcal{N}^\sigma_{i_2}(\xi^{b+k-1}) & \cdots & \mathcal{N}^\sigma_{i_k}(\xi^{b+k-1}) 
		\end{pmatrix}_{k\times k}\\[1em]
		\notag&= \begin{pmatrix} 
			(\mathcal{N}^\sigma_{i_1}(\xi))^{b} & (\mathcal{N}^\sigma_{i_2}(\xi))^{b} & \cdots & (\mathcal{N}^\sigma_{i_k}(\xi))^{b} \\ 
			(\mathcal{N}^\sigma_{i_1}(\xi))^{b+1} & (\mathcal{N}^\sigma_{i_2}(\xi))^{b+1} & \cdots & (\mathcal{N}^\sigma_{i_k}(\xi))^{b+1} \\ 
			\vdots & \vdots & \ddots & \vdots \\ 
			(\mathcal{N}^\sigma_{i_1}(\xi))^{b+k-1} & (\mathcal{N}^\sigma_{i_2}(\xi))^{b+k-1} & \cdots & (\mathcal{N}^\sigma_{i_k}(\xi))^{b+k-1}  
		\end{pmatrix}_{k\times k}.\\
		\notag&= \begin{pmatrix} 
			1 & 1 & \cdots & 1 \\ 
			\mathcal{N}^\sigma_{i_1}(\xi) & \mathcal{N}^\sigma_{i_2}(\xi) & \cdots & \mathcal{N}^\sigma_{i_k}(\xi) \\ 
			\vdots & \vdots & \ddots & \vdots \\ 
			(\mathcal{N}^\sigma_{i_1}(\xi))^{k-1} & (\mathcal{N}^\sigma_{i_2}(\xi))^{k-1} & \cdots & (\mathcal{N}^\sigma_{i_k}(\xi))^{k-1} 
		\end{pmatrix}_{k\times k} 
		\cdot  \begin{pmatrix}
		    (\mathcal{N}^\sigma_{i_1}(\xi))^b&0&0&\cdots&0\\
            0&(\mathcal{N}^\sigma_{i_2}(\xi))^b&0&\cdots&0\\
            \vdots&\vdots&\vdots&\ddots&\vdots\\
            0&0&0&\cdots&(\mathcal{N}^\sigma_{i_k}(\xi))^b
		\end{pmatrix}_{k\times k} \\[0.5em]
		\notag&= V^{\mathrm{id}}_k(\mathcal{N}^\sigma_{i_1}(\xi),\mathcal{N}^\sigma_{i_2}(\xi),\dots,\mathcal{N}^\sigma_{i_k}(\xi))^T
		\cdot\mathrm{diag}\!\left( (\mathcal{N}^\sigma_{i_1}(\xi))^b,\; (\mathcal{N}^\sigma_{i_2}(\xi))^b,\; \dots,\; (\mathcal{N}^\sigma_{i_k}(\xi))^b \right)_{k\times k}.\\
        \end{align*}
        
		This yields the determinant of the above matrix
        \begin{eqnarray}
            \det (V^\sigma_n(a_1,a_2,\dots,a_k;E') )&=& \prod_{1\leq i_j\neq i_l\leq k}\Big(\mathcal{N}^\sigma_{i_j}(\xi)-\mathcal{N}^\sigma_{i_l}(\xi)\Big)	\cdot \left( \prod_{j=1}^{k} \mathcal{N}^\sigma_{i_j}(\xi) \right)^{b}.
        \end{eqnarray}

 Following the step of the proof  \cite[Theorem~4]{Bhaintwal2012Skew}, we obtain
	\[
	\mathcal{N}^\sigma_{i_j}(\xi) \bmod p \;\neq\; \mathcal{N}^\sigma_{i_l}(\xi) \bmod p \quad \text{for all distinct } i_j,i_l \in E.
	\]
    
	Hence, $\Big(\mathcal{N}^\sigma_{i_j}(\xi)-\mathcal{N}^\sigma_{i_l}(\xi)\Big)$   $\in \mathcal{U}(GR(p^s,p^{sne}))$, and therefore the determinant of the Vandermonde matrix itself is a unit. Moreover, since $\xi$ is an unit, it follows that every $\mathcal{N}^\sigma_{i_j}(\xi)$ is also an unit. Thus, 
	$
	\det\!\left(V^\sigma_n(a_1,a_2,\dots,a_k;E') \right),
	$
	is a unit in $\mathcal{U}(GR(p^s,p^{msl})$. Since this conclusion holds for any choice of $k$ columns from $E$, the matrix $M$ is an MDS matrix.
\end{proof}

\begin{corollary}Let $a_j = c\cdot\xi^{j + b - 1}$ for $j = 1, \dots, k$, where $b$ is an integer and $c\in \mathcal{U}(GR(p^s,p^{sm}))$. 
Suppose that $a_1, \dots, a_k$ are the right roots of $g(X)$. Then, for an integer 
$t \geq k$, the matrix 
\[
M = C_g^{[t-1]} C_g^{[t-2]} \cdots C_g
\] 
is an \textup{MDS} matrix.  \end{corollary}

It is important to note that when we work over Galois rings, we obtain a larger class of quasi-recursive MDS matrices. In particular, employing such matrices in a diffusion layer may increase the resistance to algebraic attacks compared to conventional linear diffusion layers defined over finite fields (cf.; \cite{Kesarwani2021}). 

\begin{corollary}\label{Corollary 1003}
    Let $a_j = c\cdot\xi^{j + b - 1}+\eta_j$ for $j = 1, \dots, k$, where $b$ is an integer, $c\in \mathcal{U}(GR(p^s,p^{sm}))$ and $\eta_j\in \mathcal{\mathcal{N}}(GR(p^s,p^sm))$.  
Suppose that $a_1, \dots, a_k$ are the right roots of $g(X)$. Then, for an integer 
$t \geq k$, the matrix 
\[
M = C_g^{[t-1]} C_g^{[t-2]} \cdots C_g
\] 

is an \textup{MDS} matrix.
\end{corollary}
\begin{example}
Let 
\(
GR(2^2, 2^{16}) = \mathbb{Z}_4[Y]/\langle p(Y)\rangle, \quad 
p(Y) = Y^8 + Y^7 + Y^6 + Y + 1,
\)
and  $\xi = Y + \langle p(Y)\rangle$ be a primitive root of $p(Y)$ in some extension of $\mathbb{F}_2$.  
Set $a_j = \xi^{j-1}$ for $j = 1, 2, 3, 4$, and consider the polynomial
\(
g(X) = (X - a_1)(X - a_2)(X - a_3)\in GR(2^2,2^{16})[X].
\)
In Example~3 of~\cite{Kesarwani2021}, it was shown that $C^3_g$ is an MDS matrix. 
Now, consider the polynomial
\(
h(X) = (X - c_1)(X - c_2)(X - c_3),
\)
where we set $c_j = \xi^{j-1} + \eta_j$ with $\eta_j \in \mathcal{N}\big(GR(2^2, 2^{16})\big)$. 
Then, by Corollary~\ref{Corollary 1003}, the matrix
\(
M = C_h^3
\)
is an \textup{MDS} matrix over $GR(2^2, 2^{16})$.

In particular, if we take $\eta_1 = 2$, $\eta_2 = 2\xi$, and $\eta_3 = 2\xi^2$, then the polynomial
\(
h(X) = (X - c_1)(X - c_2)(X - c_3)
\)
yields a recursive MDS matrix. 
\end{example}

\begin{remark}
    The above example shows that Corollary~\ref{Corollary 1003} allows us to generate new MDS matrices simply by adding nilpotent perturbations to the roots of an existing polynomial whose associated companion matrix is MDS.  
Thus, from the original $4$-MDS matrix of Example~3, we can obtain a whole family of MDS matrices parameterized by the choice of nilpotent elements $\eta_1, \eta_2, \eta_3,\eta_4$.
\end{remark}

We now present the following lemma, formulated in light of \cite[Theorem~15]{Ali2024block}:

\begin{lemma}\label{LEMMA3}
Let 
\(
g(X)=g_0+g_1X+\cdots+g_{k-1}X^{k-1}+X^k \;\in\; \mathcal{R}[X;\sigma],
\)
and let
\(
h(X)=(g_0+\eta_0)+(g_1+\eta_1)X+\cdots+(g_{k-1}+\eta_{k-1})X^{k-1}+X^k,
\)
where $\eta_i$ for $0\leq i \leq k-1$ are nilpotent elements of $GR(p^s,p^{sm})$. Then, the matrix
\(
M' = C_h^{[t-1]} C_h^{[t-2]} \cdots C_h
\)
is an \textup{MDS} matrix whenever
\(
M = C_g^{[t-1]} C_g^{[t-2]} \cdots C_g
\)
is \textup{MDS}.
\end{lemma}
\begin{proof}
Consider
\(
g(X)=g_0+g_1X+\cdots+g_{k-1}X^{k-1}+X^k,~ 
h(X)=g_0+\eta_0+(g_1+\eta_1)X+\cdots+(g_{k-1}+\eta_{k-1})X^{k-1}+X^k.
\)
The companion matrix of $h(X)$ is
\[
C_h=\begin{bmatrix}
0 & 1 & 0 & \cdots & 0\\
0 & 0 & 1 & \cdots & 0\\
\vdots & \vdots & \vdots & \ddots & \vdots\\
g_0+\eta_0 & g_1+\eta_1 & g_2+\eta_2 & \cdots & g_{k-1}+\eta_{k-1}
\end{bmatrix}_{k\times k}.
\]

This can be decomposed as
\[
C_h = 
\begin{bmatrix}
0 & 1 & 0 & \cdots & 0\\
0 & 0 & 1 & \cdots & 0\\
\vdots & \vdots & \vdots & \ddots & \vdots\\
g_0 & g_1 & g_2 & \cdots & g_{k-1}
\end{bmatrix}_{k\times k}
+
\begin{bmatrix}
0 & 0 & 0 & \cdots & 0\\
0 & 0 & 0 & \cdots & 0\\
\vdots & \vdots & \vdots & \ddots & \vdots\\
\eta_0 & \eta_1 & \eta_2 & \cdots & \eta_{k-1}
\end{bmatrix}_{k\times k},
\]

that is,
\[
C_h = C_g + F,
\]

where $F=\begin{bmatrix}
0 & 0 & 0 & \cdots & 0\\
0 & 0 & 0 & \cdots & 0\\
\vdots & \vdots & \vdots & \ddots & \vdots\\
\eta_0 & \eta_1 & \eta_2 & \cdots & \eta_{k-1}
\end{bmatrix}_{k\times k} \in \mathcal{M}_k\!\Big(\mathcal{N}(\mathcal{R})\Big)$. Now compute
\begin{align}\label{eq:Mh}
C_h^{[t-1]} C_h^{[t-2]}\cdots C_h
&=(C_g+F)^{[t-1]}(C_g+F)^{[t-2]}\cdots(C_g+F) \notag\\
&= C_g^{[t-1]}C_g^{[t-2]}\cdots C_g 
+\underbrace{ \sum_{p=0}^{t-1}\!\Big(C_g^{[t-1]}\cdots C_g^{[p+1]}F^{[p]}C_g^{[p-1]}\cdots C_g\Big)+} \notag\\
&\quad \underbrace{ \sum_{0\le p<q\le t-1}\!\Big(C_g^{[t-1]}\cdots F^{[q]}\cdots F^{[p]}\cdots C_g\Big) 
+ \cdots + F^{[t-1]}F^{[t-2]}\cdots F}. 
\end{align}

Since $\mathcal{M}_k(\mathcal{N}(\mathcal{R}))$ is an ideal of $\mathcal{M}_k(\mathcal{R})$, the additional terms in \eqref{eq:Mh} lie inside this ideal. Thus, we can conclude that
\[
C_h^{[t-1]} C_h^{[t-2]}\cdots C_h = C_g^{[t-1]} C_g^{[t-2]}\cdots C_g + T,
\]

where $T\in \mathcal{M}_k(\mathcal{N}(\mathcal{R}))$.  Finally, by \cite[Theorem 15]{Ali2024block}, the perturbation by nilpotent elements does not affect the MDS property. Hence, $C_h^{[t-1]} C_h^{[t-2]}\cdots C_h$ is also an \textup{MDS} matrix.
\end{proof}
We now highlight an important consequence of the previous results.
\begin{remark}
Theorem~\ref{Theorem 7} provides a construction of quasi-recursive MDS matrices over Galois rings. 
Lemma~\ref{LEMMA3} ensures that, once such a matrix is obtained, 
we can generate further quasi-recursive MDS matrices by perturbing the coefficients of the defining polynomial with nilpotent elements. 
Thus, the lemma effectively reduces the search space for new constructions, since one representative MDS matrix can lead to a family of new MDS matrices.
\end{remark}

 \begin{theorem}\label{Theorem 9}
     Let $a_j = \xi^{j-1} \quad \text{for } j =  1,2,\ldots,k-1,$ and $a_k=\xi^k$.  
Suppose that $a_1, \dots, a_k$ are the right roots of $g(X)$. 	Then, for an integer $t\geq k$, the matrix 
 	$
 	M = C^{[t-1]}_g \cdot C^{[1]}_g \cdot C_g
 	$
 	is MDS if and only if $\sum_{j=1}^{k}\mathcal{N}^\sigma_{i_j}(\xi)\in \mathcal{U}(GR(p^s,p^{sne}),$ for all $\{i_1,i_2,\dots,i_k\} \subset E$.
    \end{theorem}
    \begin{proof}
       We take $a_j=\xi^{\,j-1}$ for $1\leq j \leq k-1$ and $a_k=\xi^k$, with $E'=\{i_1,i_2,\dots,i_k\}\subset E$. In view of (\ref{EQN12}), we obtain
\begin{eqnarray*}
    V^\sigma_n(a_1,a_2,\dots,a_k;E')&=&\begin{pmatrix}
			\mathcal{N}^\sigma_{i_1}(1) & \mathcal{N}^\sigma_{i_2}(1)& \cdots & \mathcal{N}^\sigma_{i_k}(1) \\
			\mathcal{N}^\sigma_{i_1}(\xi) & \mathcal{N}^\sigma_{i_2}(\xi)&\cdots & \mathcal{N}^\sigma_{i_k}(\xi) \\
			\vdots &\vdots& \ddots & \vdots \\
			\mathcal{N}^\sigma_{i_1}(\xi^{{k-2}}) &\mathcal{N}^\sigma_{i_2}(\xi^{{k-2}})& \cdots & \mathcal{N}^\sigma_{i_k}(\xi^{{k-2}})\\
            \mathcal{N}^\sigma_{i_1}(\xi^{{k}}) &\mathcal{N}^\sigma_{i_2}(\xi^{{k}})& \cdots & \mathcal{N}^\sigma_{i_k}(\xi^{{k}})
		\end{pmatrix}_{k\times k}.\\
\end{eqnarray*}
        \begin{eqnarray*}
		&=&	\begin{pmatrix}
		1 & 1& \cdots & 	1\\
			(\mathcal{N}^\sigma_{i_1}(\xi)) & (\mathcal{N}^\sigma_{i_2}(\xi))&\cdots & (\mathcal{N}^\sigma_{i_k}(\xi)) \\
			\vdots &\vdots& \ddots & \vdots \\
			(\mathcal{N}^\sigma_{i_1}(\xi))^{{k-2}} &(\mathcal{N}^\sigma_{i_2}(\xi))^{{k-2}}& \cdots & (\mathcal{N}^\sigma_{i_k}(\xi))^{{k-2}}\\
            (\mathcal{N}^\sigma_{i_1}(\xi))^{{k}} &(\mathcal{N}^\sigma_{i_2}(\xi))^{{k}}& \cdots & (\mathcal{N}^\sigma_{i_k}(\xi))^{{k}}
		\end{pmatrix}_{k\times k}.
        \end{eqnarray*}
        
        Assume $a_{j}=\mathcal{N}^\sigma_{i_j}(\xi)$ for $1\leq j \leq k$, and by using Corollary \ref{Corol3}, we obtain 
        \begin{eqnarray*}
            \det( V^\sigma_n(a_1,a_2,\dots,a_k;E'))=\prod_{1\leq l<j\leq k}(\mathcal{N}^\sigma_{i_j}(\xi)- \mathcal{N}^\sigma_{i_l}(\xi))\cdot \sum_{j=1}^{k}\mathcal{N}^\sigma_{i_j}(\xi)\in \mathcal{U}(GR(p^s,p^{sne})).
        \end{eqnarray*}
        
        Hence, $\det( V^\sigma_n(a_1,a_2,\dots,a_k;E'))\in \mathcal{U}(GR(p^s,p^{sne}))$ if and only if $\sum_{j=1}^{k}\mathcal{N}^\sigma_{i_j}(\xi)\in \mathcal{U}((GR(p^s,p^{sne})).$
    \end{proof}
    
    Similarly, we obtain the following lemma.
    \begin{lemma}
        Let $a_1=1$ and $a_j=\xi^j,$ $2\leq j \leq k.$ Suppose that $a_1, \dots, a_k$ are the right roots of $g(X)$. Then, for an integer $t\geq k$, the matrix $
 	M = C^{[t-1]}_g \cdot C^{[1]}_g \cdot C_g
 	$ is MDS if and only if $\sum_{j=1}^{k}\mathcal{N}^\sigma_{i_j}(\xi^{-1}) \in \mathcal{U}(GR(p^s,p^{sne}),$ for all $\{i_1,i_2,\dots,i_k\}\subset E.$ 
    \end{lemma}
    \begin{proof}
       Define $\tau_i=a_{k-i+1}=(\xi^{-1})^{\,i-1}c$ for $1\leq i \leq k-1$ and $\tau_k=a_1=(\xi^{-1})^{\,k}c$, where $c=\xi^k$. Then, by Theorem~\ref{Theorem 9}, the matrix 
\(
M = C^{[t-1]}_g \cdot C^{[1]}_g \cdot C_g
\)
is MDS if and only if 
\(
\sum_{j=1}^{k} \mathcal{N}^\sigma_{i_j}(\xi^{-1})\in U\big(GR(p^s,p^{sne})\big),
\)
for every subset $\{i_1,i_2,\dots,i_k\}\subset E$. This completes the proof.
    \end{proof}
    \begin{theorem}
        Let $a_1=1$ and $a_i=\xi^i$, for $2\leq i \leq k-1$, and $a_k=\xi^{k+1}.$ Suppose that $a_1, \dots, a_k$ are the right roots of $g(X)$. Then, for an integer $t\geq k$, the matrix $
 	M = C^{[t-1]}_g \cdot C^{[1]}_g \cdot C_g
 	$ is MDS if and only if $\Big(\sum_{j=1}^{k}\mathcal{N}^\sigma_{i_j}(\xi)\Big)\Big(\sum_{j=1}^{k}\mathcal{N}^\sigma_{i_j}(\xi^{-1})\Big)-1\in \mathcal{U}(GR(p^s,p^{sne}),$ for all $\{i_1,i_2,\dots,i_k\}\subset E.$ 
    \end{theorem}
    \begin{proof}
        We have $a_1=1$ and $a_i=\xi^i$ for $2\leq i \leq k-1$ and $a_k=\xi^{k+1}.$ As established in Theorem~\ref{Theroem 5}, the matrix $M$ is \textup{MDS} precisely when every $k \times k$ submatrix of the generalized Vandermonde matrix $V^\sigma_n(a_1,a_2,\dots,a_k;E)$ is nonsingular. So for any $E'=\{i_1,i_2,\dots, i_k\}\subset E,$ we have 
        \begin{eqnarray*}
            V^\sigma_n(a_1,a_2,\dots,a_k;E')&=&\begin{pmatrix}
			\mathcal{N}^\sigma_{i_1}(1) & \mathcal{N}^\sigma_{i_2}(1)& \cdots & \mathcal{N}^\sigma_{i_k}(1) \\
			\mathcal{N}^\sigma_{i_1}(\xi^2) & \mathcal{N}^\sigma_{i_2}(\xi^2)&\cdots & \mathcal{N}^\sigma_{i_k}(\xi^2) \\
			\vdots &\vdots& \ddots & \vdots \\
            \mathcal{N}^\sigma_{i_1}(\xi^{{k-1}}) &\mathcal{N}^\sigma_{i_2}(\xi^{{k-1}})& \cdots & \mathcal{N}^\sigma_{i_k}(\xi^{{k-1}})\\
			\mathcal{N}^\sigma_{i_1}(\xi^{{k+1}}) &\mathcal{N}^\sigma_{i_2}(\xi^{{k+1}})& \cdots & \mathcal{N}^\sigma_{i_k}(\xi^{{k+1}})
		\end{pmatrix}_{k\times k}.\\
		&=&	\begin{pmatrix}
		1 & 1& \cdots & 	1\\
			(\mathcal{N}^\sigma_{i_1}(\xi))^2 & (\mathcal{N}^\sigma_{i_2}(\xi))^2&\cdots & (\mathcal{N}^\sigma_{i_k}(\xi))^2 \\
			\vdots &\vdots& \ddots & \vdots \\
			(\mathcal{N}^\sigma_{i_1}(\xi))^{{k-1}} &(\mathcal{N}^\sigma_{i_2}(\xi))^{{k-1}}& \cdots & (\mathcal{N}^\sigma_{i_k}(\xi))^{{k-1}}\\
            (\mathcal{N}^\sigma_{i_1}(\xi))^{{k+2}} &(\mathcal{N}^\sigma_{i_2}(\xi))^{{k+2}}& \cdots & (\mathcal{N}^\sigma_{i_k}(\xi))^{{k+2}}
		\end{pmatrix}_{k\times k}.
        \end{eqnarray*}
        
        Application of Corollary \ref{Corol5} yields
        \begin{eqnarray*}
            \det(V^\sigma_n(a_1,a_2,\dots,a_k;E'))&=& \prod_{1\leq l< j\leq k}(\mathcal{N}^\sigma_{i_j}(\xi)-\mathcal{N}^\sigma_{i_l}(\xi))\prod_{j=1}^{k}\mathcal{N}^\sigma_{i_j}(\xi)\Bigg[\Bigg(\sum_{i=1}^{k}\mathcal{N}^\sigma_{i_j}(\xi)\Bigg)\Bigg(\sum_{i=1}^{k}\mathcal{N}^\sigma_{i_j}(\xi^{-1})\Bigg)-1\Bigg].
        \end{eqnarray*}
        
        This implies that $ \det(V^\sigma_n(a_1,a_2,\dots,a_k;I))\in \mathcal{U}(GR(p^s,p^{sne}))$ if and only if $\Big(\sum_{j=1}^{k}\mathcal{N}^\sigma_{i_j}(\xi)\Big)\Big(\sum_{j=1}^{k}\mathcal{N}^\sigma_{i_j}(\xi^{-1})\Big)-1\in \mathcal{U}(GR(p^s,p^{sne}),$ for all $\{i_1,i_2,\dots,i_k\}\subset E.$ 
    \end{proof}
\begin{definition}
	Let $\mathcal{R}$ be a finite commutative with unity and char($\mathcal{R}$)=$p$ for some prime $p$. Let $\textbf{h }= (h_0, h_1, \ldots, h_{k-1})$ be an $k$-tuple over $\mathcal{R}$.  
	Then, the linearized matrix $U(\textbf{h})$ corresponding to $\textbf{h}$ is given by
	\[
	U(\textbf{h}) = \bigl(h_i^{\,p^j}\bigr)_{i,j=0}^{\,k-1}.
	\]
    
	Explicitly,
	\[
	U(\textbf{h}) =
	\begin{pmatrix}
		h_0      & h_0^{p}      & h_0^{p^2}      & \cdots & h_0^{p^{\,k-1}} \\
		h_1      & h_1^{p}      & h_1^{p^2}      & \cdots & h_1^{p^{\,k-1}} \\
		h_2      & h_2^{p}      & h_2^{p^2}      & \cdots & h_2^{p^{\,k-1}} \\
		\vdots   & \vdots       & \vdots         & \ddots & \vdots \\
		h_{n-1}  & h_{n-1}^{p}  & h_{n-1}^{p^2}  & \cdots & h_{n-1}^{p^{\,k-1}}
	\end{pmatrix}_{k\times k}.
	\]
\end{definition}

\begin{proposition}\label{kesarwani}\cite[Proposition 1]{Kesarwani2021}
Let $\mathcal{R}$ be a finite commutative with unity and char(R)=p.  	Let $\textbf{h} = (h_0, h_1, \ldots, h_{k-1})$ be an $k$-tuple over $\mathcal{R}$.  Then,
	\[
	\det(U(\textbf{h})) = h_0 \prod_{j=0}^{k-2} \ \prod_{c_1,c_2,\ldots,c_j \in \mathbb{Z}_p} 
	\left( h_{j+1} - \sum_{i=1}^j c_i h_i \right).
	\]
\end{proposition}
If we set $s=1$ in $GR(p^s, p^{sm})$, we obtain the finite field $\mathbb{F}_{p^m}$. The following results are derived under this condition:

 \begin{theorem}
 	Let $a_j = \xi^{p^j} \quad \text{for } j =0,  1,2,\ldots,k-1.$ 
Suppose that $a_1, \dots, a_k$ are the right roots of $g(X)$. 	Then, the matrix 
 	$
 	M = C^{[t-1]}_g \cdot C^{[1]}_g \cdot C_g
 	$
 	is MDS if and only if both of the following conditions hold:
 	\begin{itemize}
 		\item[\textnormal{(i)}] $a_i-a_j\neq 0$ for all $1\leq i\neq j \leq k.$ 
 		\item[\textnormal{(ii)}]
 		\(
 	\sum_{l=1}^{k} b_i \mathcal{N}^\sigma_{i_l}(\lambda) \neq 0 \quad 
 	\text{for } (b_1, b_2, \ldots, b_k) \neq (0,0,\ldots,0) \in \mathbb{Z}_p^{k}, \quad 
 		\text{for all sets } Z = \{i_1, i_2, \ldots, i_k\} \subset E,
 		\)
 		where $(0,0,\ldots,0)\neq(c_1, c_2, \ldots, c_k)\in \mathbb{Z}^k_p $ is a non-zero vector in $\mathbb{Z}_p^k$.
 	\end{itemize}
 \end{theorem}

\begin{proof}
	By setting $a_j = \xi^{p^j}$ for $j = 0,1,2,\dots,k-1$ in (\ref{EQN12}), we obtain

	\begin{eqnarray*}
			V^\sigma_n(a_1,a_2,\dots,a_k;E') &=& 
		\begin{pmatrix}
			\mathcal{N}^\sigma_{i_1}(\xi) & \mathcal{N}^\sigma_{i_2}(\xi)& \cdots & \mathcal{N}^\sigma_{i_k}(\xi) \\
			\mathcal{N}^\sigma_{i_1}(\xi^{p}) & \mathcal{N}^\sigma_{i_2}(\xi^{p})&\cdots & \mathcal{N}^\sigma_{i_k}(\xi^{p}) \\
			\vdots &\vdots& \ddots & \vdots \\
			\mathcal{N}^\sigma_{i_1}(\xi^{p^{k-1}}) &\mathcal{N}^\sigma_{i_2}(\xi^{p^{k-1}})& \cdots & \mathcal{N}^\sigma_{i_k}(\xi^{p^{k-1}})
		\end{pmatrix}_{k \times k}\\
		&=&	\begin{pmatrix}
			(\mathcal{N}^\sigma_{i_1}(\xi)) & 	(\mathcal{N}^\sigma_{i_2}(\xi))& \cdots & 	(\mathcal{N}^\sigma_{i_k}(\xi)) \\
				(\mathcal{N}^\sigma_{i_1}(\xi))^{p} & (\mathcal{N}^\sigma_{i_2}(\xi))^{p}&\cdots & (\mathcal{N}^\sigma_{i_k}(\xi))^{p} \\
			\vdots &\vdots& \ddots & \vdots \\
			(\mathcal{N}^\sigma_{i_1}(\xi))^{p^{k-1}} &(\mathcal{N}^\sigma_{i_2}(\xi))^{p^{k-1}}& \cdots & (\mathcal{N}^\sigma_{i_k}(\xi))^{p^{k-1}}
		\end{pmatrix}_{k \times k}.
	\end{eqnarray*}
    
We observe that $V^\sigma_n(a_1,a_2,\dots,a_k;E')=(U(\textbf{h}))^{T}$, where $\textbf{h}=(\mathcal{N}^\sigma_{i_1}(\xi), \mathcal{N}^\sigma_{i_2}(\xi),\dots,\mathcal{N}^\sigma_{i_k}(\xi)).$ Now from Proposition \ref{kesarwani}, $U(\textbf{h})$ is non singular if and only if

	\begin{eqnarray}\label{EQN4}
		\det(	V^\sigma_n(a_1,a_2,\dots,a_k;E'))\notag&=&\Bigg(\mathcal{N}^\sigma_{i_1}(\xi))\prod_{j=1}^{k-1}\prod_{c_1,c_2,\dots,c_j\in \mathbb{Z}_p}(\mathcal{N}^\sigma_{i_{j+1}}(\xi)-\sum_{l=1}^{j}c_i(\mathcal{N}^\sigma_{i_l}(\xi))^p\Bigg) \textup{ is invertible in } \mathbb{F}_{p^{ne}}.\\
		&~&.
	\end{eqnarray}
    
 Now observe that the terms appearing in the products, when the quantity in Equation (\ref{EQN4}) is 
multiplied by $b$ for all $b \in \mathbb{Z}_p^{\ast}$, are exactly
\[
\sum_{l=1}^{k} b_i \mathcal{N}^\sigma_{i_l}(\lambda) \neq 0 \quad 
\text{for } (b_1, b_2, \ldots, b_k) \neq (0,0,\ldots,0) \in \mathbb{Z}_p^{k}.
\]

Hence the result.
 
\end{proof}
\begin{corollary}
		Let 
	\(
	a_1=1,~a_j = \xi^{p^{j-1}} \quad \text{for } j = 2,3,\ldots,k.
	\) are the right roots of $g(X).$
	Then, the matrix 
	\(
	M = C^{[t-1]}_g\cdot C^{[1]}_g\cdot C_g
	\)
	is MDS  if and only if both of the following conditions hold:
	\begin{enumerate}
		\item[\textnormal{(i)}] $a_i-a_j\neq 0$ for all $1\leq i\neq j \leq k.$ 
		\item[\textnormal{(ii)}] 
		\(
		\sum_{i=1}^k c_i \mathcal{N}^\sigma_{i_r}(\xi) \neq 0, \quad 
		\text{for all sets } Z = \{i_1, i_2, \ldots, i_k\},
		\subset E\)
		where $(c_2, \ldots, c_k) \neq (0,0,\ldots,0)$ is a non-zero vector in $\mathbb{Z}_p^{k-1}$, and $c_1=-\sum_{j=2}^{k}c_j$. 
	\end{enumerate}
\end{corollary}
\begin{proof}
	 By setting $a_1=1, ~a_j = \xi^{p^{j-2}}$ for $j = 2,3,\dots,k$ in (\ref{EQN12}). For an arbitrary choice of 
	$k$ column indices $\{i_1, \dots, i_k\} \subset E$, we obtain
	\begin{eqnarray*}
		V^\sigma_n(a_1,a_2,\dots,a_k;E') &=& 
		\begin{pmatrix}
			\mathcal{N}^\sigma_{i_1}(1) & \mathcal{N}^\sigma_{i_2}(1)& \cdots & \mathcal{N}^\sigma_{i_k}(1) \\
			\mathcal{N}^\sigma_{i_1}(\xi) & \mathcal{N}^\sigma_{i_2}(\xi)&\cdots & \mathcal{N}^\sigma_{i_k}(\xi) \\
			\vdots &\vdots& \ddots & \vdots \\
			\mathcal{N}^\sigma_{i_1}(\xi^{p^{k-1}}) &\mathcal{N}^\sigma_{i_2}(\xi^{p^{k-2}})& \cdots & \mathcal{N}^\sigma_{i_k}(\xi^{p^{k-2}})
		\end{pmatrix}_{k\times k}.\\
		&=&	\begin{pmatrix}
		1 & 1& \cdots & 	1\\
			(\mathcal{N}^\sigma_{i_1}(\xi)) & (\mathcal{N}^\sigma_{i_2}(\xi))&\cdots & (\mathcal{N}^\sigma_{i_k}(\xi)) \\
			\vdots &\vdots& \ddots & \vdots \\
			(\mathcal{N}^\sigma_{i_1}(\xi))^{p^{k-2}} &(\mathcal{N}^\sigma_{i_2}(\xi))^{p^{k-2}}& \cdots & (\mathcal{N}^\sigma_{i_k}(\xi))^{p^{k-2}}
		\end{pmatrix}_{k\times k}.
	\end{eqnarray*}
    
	By subtracting first column from all other columns, and taking $z_{i_r}=(\mathcal{N}^\sigma_{i_r}(\xi)-\mathcal{N}^\sigma_{i_1}(\xi)),$ for $2\leq i\leq k,$ we obtain
	\begin{eqnarray*}
	V^{'\sigma}_n&=&
	\begin{pmatrix}
			1&0&0&\cdots&0\\
			(\mathcal{N}^\sigma_{i_1}(\xi))&z_{i_2}&z_{i_3}&\cdots&z_{i_k}\\
		(\mathcal{N}^\sigma_{i_1}(\xi))^{p}&z^p_{i_2}&z^p_{i_3}&\cdots&z^p_{i_k}\\
		\vdots&\vdots&\vdots&\ddots&\vdots\\
		(\mathcal{N}^\sigma_{i_1}(\xi))^{p^{k-2}}&z^{p^{k-2}}_{i_1}&z^{p^{k-2}}_{i_3}&\cdots&z^{p^{k-2}}_{i_k}\\
		\end{pmatrix}_{k\times k}.
	\end{eqnarray*}
    
	This implies that 
	\begin{eqnarray*}
		\det V^{'\sigma}_n&=&\det U(z_{i_2},z_{i_3},\dots,z_{i_k})
		=  z_{i_2} \prod_{j=2}^{k-1} \ \prod_{c_1,c_2,\ldots,c_j \in \mathbb{Z}_p} 
		\left( z_{i_{j+1}} - \sum_{l=1}^j c_l z_{i_l} \right)		\\
		&=&(\mathcal{N}^{\sigma}_{i_2}(\xi)-\mathcal{N}^{\sigma}_{i_1}(\xi))\prod_{j=2}^{k-1}\prod_{c_1,c_2,\dots,c_j\in \mathbb{Z}_p}\Bigg((\mathcal{N}^{\sigma}_{i_{j+1}}(\xi)-\mathcal{N}^{\sigma}_{i_1}(\xi))-\sum_{l=1}^{j}c_l(\mathcal{N}^{\sigma}_{i_l}(\xi)-\mathcal{N}^{\sigma}_{i_1}(\xi))\Bigg).
	\end{eqnarray*}
    
Thus, one can see that $\det U(z_{i_2},z_{i_3},\dots,z_{i_k})$ is a unit element if and only if $\sum_{l=2}{k}c_lz_{i_l}\neq 0$ for all $\{i_2,\dots,i_k\}\subset E$. Note that $\sum_{l=2}^{k}c_lz_{i_l}=\sum_{l=2}^{k}(\mathcal{N}^\sigma_{i_l}(\xi)-\mathcal{N}^\sigma_{i_1}(\xi^p))=\sum_{l=1}^{k}c_l\mathcal{N}^\sigma_{i_l}(\xi)$, where $c_1=-\sum_{l=2}^{k}c_l.$ Hence the result.

\end{proof}

\section{Conclusion}\label{Sec:Conclusion}

In this paper, we have presented a systematic method for the direct construction of Maximum Distance Separable (MDS) matrices over Galois rings by establishing a direct link between the generator polynomials of skew cyclic codes and their associated companion matrices. This work extends the rich theory of MDS matrices from finite fields to the more general and cryptographically relevant setting of Galois rings \(GR(p^{s}, p^{sm})\). Our results not only generalize prior work over finite fields but also hold practical significance for implementations over \(\mathbb{F}_{2^m}\) (i.e., when \(s=1\) and \(p=2\)), where efficient diffusion layers are crucial in lightweight cryptographic designs. Building on this foundation, future research may focus on the direction on analyzing the the cost of such MDS matrices construction, optimizing hardware implementations, and investigating applications in block ciphers and hash functions with enhanced security and efficiency.

%

\medskip
\bibliographystyle{abbrv}
\bibliography{Ref}

\end{document}